\documentclass[reqno,10pt,a4paper]{amsart}

\usepackage{amssymb,mathptmx,eucal,mathrsfs,array,setspace,color,geometry,enumitem}
\usepackage[noadjust]{cite}              
\usepackage{graphicx}
\usepackage{tikz}

\DeclareSymbolFont{largesymbols}{OMX}{zplm}{m}{n} 

\setitemize{leftmargin=*}                
\setenumerate{leftmargin=*,              
              label=\textup{\arabic*)}}  

\newcolumntype{C}{>{$}c<{$}} 

\numberwithin{equation}{section}

\allowdisplaybreaks

\renewcommand{\ge}{\geq}
\renewcommand{\le}{\leq}

\DeclareMathOperator{\ad}{ad}
\DeclareMathOperator{\vspn}{span}

\newcommand{\alg}[1]{\mathfrak{#1}} 
\newcommand{\grp}[1]{\mathsf{#1}} 

\newcommand{\func}[2]{#1 \left( #2 \right)} 
\newcommand{\tfunc}[2]{#1 \bigl( #2 \bigr)} 

\newcommand{\brac}[1]{\left( #1 \right)}
\newcommand{\tbrac}[1]{\bigl( #1 \bigr)}
\newcommand{\sqbrac}[1]{\left[ #1 \right]}

\newcommand{\set}[1]{\left\{ #1 \right\}}
\newcommand{\tset}[1]{\bigl\{ #1 \bigr\}}

\newcommand{\st}{\mspace{5mu} : \mspace{5mu}} 

\newcommand{\abs}[1]{\left\lvert #1 \right\rvert}

\newcommand{\ZZ}{\mathbb{Z}}
\newcommand{\NN}{\mathbb{N}}
\newcommand{\QQ}{\mathbb{Q}}
\newcommand{\RR}{\mathbb{R}}
\newcommand{\CC}{\mathbb{C}}

\newcommand{\pd}{\partial}     

\newcommand{\dd}{\mathrm{d}}   
\newcommand{\ii}{\mathfrak{i}} 
\newcommand{\ee}{\mathsf{e}}   

\newcommand{\wun}{\mathbf{1}}  

\newcommand{\bra}[1]{\bigl\langle #1 \bigr\rvert}
\newcommand{\ket}[1]{\bigl\lvert #1 \bigr\rangle}

\newcommand{\braket}[2]{\bigl\langle #1 \bigm\vert #2 \bigr\rangle}                  

\newcommand{\normord}[1]{\mbox{${} : #1 : {}$}} 

\newcommand{\comm}[2]{\bigl[ #1 , #2 \bigr]}

\newcommand{\inner}[2]{\bigl\langle #1 , #2 \bigr\rangle}
\newcommand{\corrfn}[1]{\bigl\langle #1 \bigr\rangle}
\newcommand{\inprod}[2]{\inner{#1}{#2}}                   
\newcommand{\intprod}[2]{\left\langle #1 , #2 \right\rangle}

\newcommand{\hv}[1]{u_{#1}}    
\newcommand{\gv}[1]{v_{#1}}    
\newcommand{\av}[1]{w_{#1}}    
\newcommand{\wv}[1]{\ket{#1}}  
\newcommand{\dwv}[1]{\bra{#1}} 

\newcommand{\gvac}{\gv{}}      

\newcommand{\Ra}{\Rightarrow}
\newcommand{\lra}{\longrightarrow}
\newcommand{\dses}[5]{0 \lra #1 \overset{#2}{\lra} #3 \overset{#4}{\lra} #5 \lra 0} 

\newcommand{\symfunc}{\Lambda}                            
\newcommand{\sympol}[1]{\Lambda_{#1}}                     
\newcommand{\monsym}[1]{\mathsf{m}_{#1}}                  
\newcommand{\fmonsym}[2]{\mathsf{m}_{#1}(#2)}             
\newcommand{\dmonsym}[2]{\mathsf{g}_{#1}^{#2}}            
\newcommand{\fdmonsym}[3]{\mathsf{g}_{#1}^{#2}\left(#3\right)}
\newcommand{\elsym}[1]{\mathsf{e}_{#1}}                   
\newcommand{\felsym}[2]{\mathsf{e}_{#1}\left(#2\right)}

\newcommand{\powsum}[1]{\mathsf{p}_{#1}}                  
\newcommand{\fpowsum}[2]{\mathsf{p}_{#1}\left(#2\right)}
\newcommand{\jack}[2]{\mathsf{P}_{#1}^{#2}}               
\newcommand{\fjack}[3]{\mathsf{P}_{#1}^{#2}\left(#3\right)}
\newcommand{\djack}[2]{\mathsf{Q}_{#1}^{#2}}              
\newcommand{\fdjack}[3]{\mathsf{Q}_{#1}^{#2}\left(#3\right)}

\newcommand{\categ}[1]{\mathscr{#1}}   

\newcommand{\affine}[1]{\widehat{#1}}
\newcommand{\SLG}[2]{\grp{#1} \bigl( #2 \bigr)}                             
\newcommand{\SLA}[2]{\alg{#1} \bigl( #2 \bigr)}                             
\newcommand{\AKMA}[2]{\affine{\alg{#1}} \left( #2 \right)}                  

\newcommand{\GA}{\alg{G}}              
\newcommand{\HA}{\alg{H}}              

\newcommand{\UEA}{\alg{U}}             

\newcommand{\VOA}[1]{\mathsf{#1}}
\newcommand{\MinMod}[2]{\VOA{M} \bigl( #1 , #2 \bigr)}    
\newcommand{\AdmMod}[2]{\VOA{A}_1 \bigl( #1 , #2 \bigr)}  
\newcommand{\UVOA}[1]{\VOA{V}_{#1}}                       

\newcommand{\zhu}[1]{\textup{Zhu} \bigl[ #1 \bigr]} 
\newcommand{\filtzhu}[2]{\textup{Zhu}_{#1} \bigl[ #2 \bigr]}

\newcommand{\Fin}[1]{\overline{#1}}

\newcommand{\conjaut}{\mathsf{w}}                     
\newcommand{\finconjaut}{\Fin{\mathsf{w}}}            
\newcommand{\conjmod}[1]{\tfunc{\conjaut}{#1}}        
\newcommand{\finconjmod}[1]{\tfunc{\finconjaut}{#1}}  

\newcommand{\ghconjaut}{\mathsf{c}}                   
\newcommand{\finghconjaut}{\Fin{\mathsf{c}}}
\newcommand{\ghconjmod}[1]{\tfunc{\ghconjaut}{#1}}    
\newcommand{\finghconjmod}[1]{\tfunc{\finghconjaut}{#1}}

\newcommand{\Mod}[1]{\mathcal{#1}}          

\newcommand{\FF}[1]{\Mod{F}_{#1}}                 
\newcommand{\GhostVac}{\Mod{G}}                   
\newcommand{\GhostRel}[1]{\Mod{G}_{#1}}           
\newcommand{\FinGhostVac}{\Fin{\Mod{G}}}          
\newcommand{\FinGhostRel}[1]{\Fin{\Mod{G}}_{#1}}  
\newcommand{\Wak}[1]{\Mod{W}_{#1}}                
\newcommand{\WakRel}[1]{\Mod{W}_{#1}}             

\newcommand{\FinIrr}[1]{\Fin{\Mod{L}}_{#1}}  
\newcommand{\FinDisc}[1]{\Fin{\Mod{D}}_{#1}} 
\newcommand{\FinRel}[1]{\Fin{\Mod{R}}_{#1}}  

\newcommand{\Ver}[1]{\Mod{V}_{#1}}          
\newcommand{\Irr}[1]{\Mod{L}_{#1}}          
\newcommand{\Disc}[1]{\Mod{D}_{#1}}         
\newcommand{\Rel}[1]{\Mod{E}_{#1}}          
\newcommand{\RelVer}[1]{\Mod{R}_{#1}}       

\newcommand{\vop}[1]{\mathscr{V}_{#1}}      
\newcommand{\scr}{\mathscr{Q}}              
\newcommand{\antiscr}{\mathscr{P}}          
\newcommand{\scrs}[1]{\scr^{[#1]}}          

\newcommand{\secref}[1]{Section~\ref{#1}}

\newcommand{\appref}[1]{Appendix~\ref{#1}}

\newcommand{\thmref}[1]{Theorem~\ref{#1}}

\newcommand{\propref}[1]{Proposition~\ref{#1}}

\newcommand{\corref}[1]{Corollary~\ref{#1}}

\newcommand{\cft}{conformal field theory}
\newcommand{\cfts}{conformal field theories}
\newcommand{\uea}{universal enveloping algebra}

\newcommand{\lcft}{logarithmic conformal field theory}
\newcommand{\lcfts}{logarithmic conformal field theories}
\newcommand{\WZW}{Wess-Zumino-Witten}
\newcommand{\ope}{operator product expansion}
\newcommand{\opes}{operator product expansions}
\newcommand{\hw}{highest weight}
\newcommand{\lw}{lowest weight}
\newcommand{\rhw}{relaxed \hw{}}
\newcommand{\hwv}{\hw{} vector}
\newcommand{\rhwv}{relaxed \hwv{}}
\newcommand{\hwvs}{\hw{} vectors}
\newcommand{\rhwvs}{relaxed \hwvs{}}
\newcommand{\sv}{singular vector}
\newcommand{\svs}{singular vectors}
\newcommand{\hwm}{\hw{} module}

\newcommand{\rhwm}{\rhw{} module}
\newcommand{\hwms}{\hw{} modules}
\newcommand{\lwms}{\lw{} modules}
\newcommand{\rhwms}{\rhw{} modules}
\newcommand{\voa}{vertex operator algebra}
\newcommand{\voas}{vertex operator algebras}

\newcommand{\rhs}{right-hand side}

\newtheoremstyle{cited}                                 
  {3pt}                                                 
  {3pt}                                                 
  {\itshape}                                            
  {}                                                    
  {\bfseries}                                           
  {.}                                                   
  {.5em}                                                
  {\thmname{#1} \thmnumber{#2} \thmnote{\normalfont#3}} 

\theoremstyle{plain}
\newtheorem{thm}{Theorem}[section]
\newtheorem{prop}[thm]{Proposition}

\newtheorem{cor}[thm]{Corollary}
\theoremstyle{definition} 

\newtheorem*{conj}{Conjecture}
\newtheorem*{defn}{Definition}
\theoremstyle{cited}
\newtheorem{cthm}[thm]{Theorem}
\newtheorem{cprop}[thm]{Proposition}

\begin{document}

\title[Jack polynomials and admissible $\AKMA{sl}{2}$]{Relaxed singular vectors, Jack symmetric functions \\ and fractional level $\AKMA{sl}{2}$ models}

\author[D Ridout]{David Ridout}

\address[David Ridout]{
Department of Theoretical Physics \\
Research School of Physics and Engineering;
and
Mathematical Sciences Institute;
Australian National University \\
Acton, ACT 2601 \\
Australia
}

\email{david.ridout@anu.edu.au}

\author[S Wood]{Simon Wood}

\address[Simon Wood]{
Department of Theoretical Physics \\
Research School of Physics and Engineering;
and
Mathematical Sciences Institute;
Australian National University \\
Acton, ACT 2601 \\
Australia
}

\email{simon.wood@anu.edu.au}

\thanks{\today}

\begin{abstract}
The fractional level models are (logarithmic) conformal field theories associated with affine Kac-Moody (super)algebras at certain levels $k \in \QQ$.  They are particularly noteworthy because of several longstanding difficulties that have only recently been resolved.  Here, Wakimoto's free field realisation is combined with the theory of Jack symmetric functions to analyse the fractional level $\AKMA{sl}{2}$ models.  The first main results are explicit formulae for the singular vectors of minimal grade in relaxed Wakimoto modules.  These are closely related to the minimal grade singular vectors in relaxed (parabolic) Verma modules.  Further results include an explicit presentation of Zhu's algebra and an elegant new proof of the classification of simple relaxed highest weight modules over the corresponding vertex operator algebra.  These results suggest that generalisations to higher rank fractional level models are now within reach.
\end{abstract}

\maketitle

\onehalfspacing

\section{Introduction} \label{sec:Intro}

We begin by briefly summarising some historical details regarding fractional level $\AKMA{sl}{2}$ models and the application of Jack symmetric functions to \cft{}.  The work presented here forms a part of of an ambitious project aimed at elucidating the properties of general fractional level models as fundamental examples of \lcfts{}.  A short digression on the notion of admissibility follows as we use the term in a non-standard manner as compared to much of the literature.  This also provides us with an opportunity to fix some notation.  Finally, we outline the main results of the research reported here.

\subsection{Fractional level models and Jack symmetric functions}

The fractional level $\AKMA{sl}{2}$ models are \cfts{} with a long and notorious history, originally proposed by Kent \cite{KenInf86} as non-unitary models whose existence would lead to a uniform coset construction for all the Virasoro minimal models.  This proposal received a significant boost from the subsequent announcement of Kac and Wakimoto \cite{KacMod88} that $\AKMA{sl}{2}$ has, for precisely the levels required for this coset construction, a finite set of simple \hwms{} whose characters close under modular transformations.  Moreover, they observed that this property persisted for higher rank affine Kac-Moody algebras.  However, computations \cite{KohFus88,BerFoc90,AwaFus92} of the fusion rules of the purported theories gave confusing and conflicting results.  In particular, substituting the modular S-matrix entries into the Verlinde formula resulted in negative fusion multiplicities.  Despite a flurry of subsequent work, no resolution was agreed upon; in their discussion, the authors of the textbook \cite{DiFCon97} suggested that the fractional level models may be ``intrinsically sick''.

The first steps towards curing this sickness were made by Gaberdiel \cite{GabFus01}, whose explicit fusion computations for the $\AKMA{sl}{2}$ model of level $k=-\tfrac{4}{3}$ demonstrated that one was forced to consider modules that were not \hw{}.  Indeed, he found that fusing the \hwms{} of Kac and Wakimoto resulted in infinite numbers of modules whose conformal weights were not bounded below, modules that were not \hw{} with respect to any Borel subalgebra (relaxed modules), and modules upon which the Virasoro mode $L_0$ acted non-semisimply.  This showed that the $k=-\tfrac{4}{3}$ model was not a rational \cft{}, as had been implicitly assumed in earlier studies, but was, in fact, logarithmic.  Subsequent works \cite{LesSU202,LesLog04,RidSL210,RidFus10} extended these results to $k=-\tfrac{1}{2}$ and thence to the closely related $\beta \gamma$ ghost systems \cite{RidSL208,RidFus10,RidBos14}.  Moreover, it was also shown \cite{RidSL208} that a misunderstanding concerning the role of convergence regions in the modular transformations of Kac and Wakimoto was to blame for the failure of the Verlinde formula.  To get a genuine action of the modular group and a working Verlinde formula, one needs to include all the simple relaxed modules and all their twists under spectral flow \cite{CreMod12,CreMod13}.

We regard these recent advances as demonstrating that the sickness of the $\AKMA{sl}{2}$ models has been cured and we expect that this cure will be just as effective for higher rank Kac-Moody algebras.  Given the logarithmic nature of the theories, it is now reasonable to expect that the fractional level models will play an important role in understanding \lcft{}, much as the non-negative integer level \WZW{} models did for rational \cft{}.  We therefore view the fractional level models as objects that deserve intense study.  In particular, one should at least determine the spectrum of simple modules, compute character formulae and verify that the modular properties of these characters lead to non-negative integers upon applying the Verlinde formula.  More ambitiously, one would like to understand the fusion ring generated by the simple modules, the structure of the indecomposable modules generated by fusion, and the corresponding three- and four-point correlation functions.  From a mathematical perspective, one can ask about homological properties of the spectrum (a category of modules over the corresponding \voa{}) including rigidity and the identification of the projective and injective modules.  We think that it is not unreasonable to expect that there are beautiful answers to all these questions, given that they concern structures built from affine Kac-Moody algebras.

While some, though not all, of these questions have been answered for $\AKMA{sl}{2}$ at levels $k=-\tfrac{4}{3}$ and $-\tfrac{1}{2}$, some of the machinery employed is clearly unfeasible for more general levels and ranks.  In our opinion, it is likely that the modular story will remain under control within the so-called standard module framework \cite{CreLog13,RidVer14}.  On the other hand, brute force methods such as the Nahm-Gaberdiel-Kausch algorithm \cite{NahQua94,GabInd96} for fusion products are already computationally-prohibitive in all but the simplest cases.  Further progress will instead require the development of more general alternatives and free field realisations seem to be the obvious candidates in this respect.  Here, we use a free field realisation to address one of the most basic questions of all, that of determining the simple modules of a fractional level theory.  We restrict ourselves to the $\AKMA{sl}{2}$ models as a testing ground, leaving the challenge of higher ranks for future publications.

The standard free field realisation of an affine Kac-Moody algebra is due to Wakimoto \cite{WakFoc86}, for $\AKMA{sl}{2}$, and Feigin and Frenkel \cite{FeiFam88} in general.  For $\AKMA{sl}{2}$, the Wakimoto realisation combines a free boson with a pair of bosonic ghosts of central charge $c=2$.  The virtue of free field theories such as these is that their underlying Lie algebras are almost abelian.  More precisely, the negative modes (creation operators) all commute among themselves, hence one can invoke symmetric group theory, in particular symmetric polynomials and functions, to analyse certain representation-theoretic questions.  For the question of classifying the simple modules of the fractional level $\AKMA{sl}{2}$ models, it turns out that the key lies in the Jack symmetric functions \cite{JacCla70}.

The relevance of Jack symmetric functions to \cft{} goes back to the work of Mimachi and Yamada \cite{MimSin95} who realised that they provide elegant expressions for the \svs{} of Verma modules over the Virasoro algebra.  Explicit \sv{} formulae are useful for many field-theoretic investigations including those of the spectrum of primary fields, the fusion rules and the correlation functions.  However, the general \sv{} formulae that were then known, for example those of \cite{BenDeg88,BauCov91}, are not particularly tractable for these purposes.  On the other hand, the Jack function formulae were derived directly from the Feigin-Fuchs free field realisation of the Virasoro algebra (also known as the Coulomb gas), as developed by Tsuchiya and Kanie \cite{TsuFoc86} and Felder \cite{FelBRST89}, which is far better suited to explicit computation.

Unfortunately, it appears that the power of symmetric function theory was not immediately exploited in \cft{} studies, perhaps because of an unfamiliarity with Jack symmetric functions (Macdonald's influential textbook \cite{MacSym95} did not appear until several years later).  However, in more recent times, symmetric function theory has been embraced by the community, particularly as a means to prove the AGT conjecture \cite{AldLio10} which relates Liouville \cfts{} to the instanton calculus of Yang-Mills theories (although it now appears that a generalisation of the Jack symmetric functions will be required \cite{MorTow14} to prove this conjecture).

The work reported here, using Jack symmetric functions to classify simple modules of the fractional level $\AKMA{sl}{2}$ models, has its genesis in \cite{TsuExt13} where this formalism was used to classify, among other things, the simple modules of a family \cite{FeiLog06} of (logarithmic) \cfts{} called the $(p_+,p_-)$ triplet models.  The methods developed for this purpose were also applied to give a far more elegant proof of the \sv{} formulae of Mimachi and Yamada.  More recently \cite{RidJac14}, it was shown that these methods lead to an elegant new proof of the classification of the Virasoro minimal model modules.  We recall that the original classification proof of Wang \cite{WanRat93} combined a projection formula for the \sv{} of the vacuum module, stated by Feigin and Fuchs and eventually proven (fifteen years later) by Astashkevich and Fuchs \cite{AstAsy97}, with some intricate cohomological arguments.

The content of this paper is then that these simplified methods generalise to the fractional level $\AKMA{sl}{2}$ models.  Specifically, we deduce explicit formulae for \svs{} in Wakimoto modules and classify the simple modules.  We take this as strong evidence that the symmetric function techniques developed here will further generalise to models over higher rank Kac-Moody algebras and superalgebras.  Moreover, it is clear that these techniques can be profitably exploited to investigate other important representation-theoretic questions including the structure theory of relaxed Verma and Wakimoto modules.  We hope to report on this in the future.

It should be emphasised that, as with the Virasoro case reported in \cite{RidJac14}, the $\AKMA{sl}{2}$ classification result is not new.  The \hw{} classification is due to Adamovi\'{c} and Milas \cite{AdaVer95} and, independently, Dong, Li and Mason \cite{DonVer97}.  However, their proofs mimic that of Wang, relying upon a projection formula for the \sv{} of the vacuum module stated (without proof) by Fuchs \cite{FucTwo89}.  From this projection formula, Adamovi\'{c} and Milas also derive a classification result for what we call, following \cite{FeiEqu98}, the \emph{relaxed} \hwms{}.\footnote{The existence of even more general weight modules follows directly from twisting by the spectral flow automorphisms of $\AKMA{sl}{2}$.  Well known in the physics literature, this twisting is an important ingredient in \cite{FeiEqu98,MalStr01} and first seems to have been explicitly noted for fractional level $\AKMA{sl}{2}$ models in \cite{GabFus01}.}  It is not clear to us if these projection formulae will generalise easily to higher ranks; the tedium of the proof in the Virasoro case alone warrants, in our opinion, the development of the rather more elegant symmetric function methods.  With this in mind, we remark that, to the best of our knowledge, there are no general classification results known for fractional level models of rank greater than $1$.

\subsection{Basic concepts and notation}

Suppose that we have a \cft{} whose chiral algebra is identified as a Lie (super)algebra $\alg{g}$, for example, the Virasoro algebra or an affine Kac-Moody (super)algebra.  More general chiral algebras can be accommodated within the formalism to follow, but this level of generality will suffice for the purposes of the article.  There is always a module of the chiral algebra, called the vacuum module, that carries the structure of a \voa{} $\VOA{V}$.  The elements of the chiral algebra are organised into fields that generate $\VOA{V}$ and the (anti)commutation relations of the chiral algebra are equivalent to the \opes{} of these generating fields.  We make the following definition:
\begin{defn}
Consider a \voa{} $\VOA{V}$ corresponding to a Lie (super)algebra $\alg{g}$ as above.  If the \opes{} of the generating fields constitute a complete set of algebraic relations, then $\VOA{V}$ is said to be \emph{universal}.
\end{defn}
\noindent The terminology comes from noting that a \voa{} is a quotient of $\VOA{V}$ if, and only if, it has a set of generating fields that satisfy the same \opes{} as those of $\VOA{V}$.

One way to understand universality is to consider what it means for the vacuum module.  When $\alg{g}$ is the Virasoro algebra, the vacuum module of the universal \voa{} is the quotient of the Verma module whose (generating) \hwv{} has conformal weight $0$ by the Verma submodule generated by the \sv{} of conformal weight $1$ (the vacuum must be translation-invariant).  For $\alg{g} = \AKMA{sl}{2}$, one quotients the Verma module generated by the \hwv{} of $\SLA{sl}{2}$-weight $0$ by the Verma submodule generated by the \sv{} of $\SLA{sl}{2}$-weight $-2$ (as required by the state-field correspondence).  In both cases, it may happen that this universal vacuum module is not simple.  But, quotienting by a non-trivial proper submodule amounts to imposing additional relations upon the \voa{}, so the result would no longer qualify as universal.

Suppose now that we have a parametrised family of universal \voas{}.  For example, the Virasoro algebra and $\AKMA{sl}{2}$ each define a one-parameter family of \voas{}, parametrised by the central charge $c \in \CC$ and the level $k \in \CC \setminus \set{-2}$, respectively.  We characterise the non-simple members of this family.
\begin{defn}
Suppose that one has a family of universal \voas{} $\set{\VOA{V}_i}_{i \in I}$, parametrised by some index set $I$.  A given value $i$ of the parameter is said to be \emph{admissible} if the corresponding universal \voa{} $\VOA{V}_i$ is not simple.
\end{defn}
\noindent For the Virasoro algebra, the admissible central charges $c$ are then precisely those for which there exist coprime integers $p,p' \in \ZZ_{\ge 2}$ satisfying $c = 13 - 6(t + t^{-1})$, where $t = \tfrac{p}{p'}$.  These central charges correspond, of course, to the Virasoro minimal models which are commonly denoted by $\MinMod{p}{p'}$.  We will also denote the simple Virasoro \voa{} of this (admissible) central charge by $\MinMod{p}{p'}$, for convenience.

For $\AKMA{sl}{2}$, the structure theory of its Verma modules \cite{KacStr79} leads to the following characterisation:
\begin{prop}
The admissible levels \(k\) of the universal \voas{} of $\AKMA{sl}{2}$ are precisely those for which there exist coprime integers \(u\in\ZZ_{\ge2}\) and \(v\in\ZZ_{\ge1}\) satisfying $k=-2+t$, where \(t=\frac{u}{v}\).
\end{prop}
\noindent Needless to say, the admissible levels of $\AKMA{sl}{2}$ are precisely those of the fractional level $\AKMA{sl}{2}$ models (including those of non-negative integer level, for convenience).  To emphasise the analogy with the Virasoro minimal models $\MinMod{p}{p'}$, we shall denote by $\AdmMod{u}{v}$ both the fractional level $\AKMA{sl}{2}$ model with $k = -2 + \tfrac{u}{v}$ and the corresponding simple \voa{}.  We regard the $\AdmMod{u}{v}$ as the minimal models of $\AKMA{sl}{2}$.  Note that $\AdmMod{k+2}{1}$ requires \(k\in\ZZ_{\ge 0}\) and is therefore just the level $k$ \WZW{} model on $\SLG{SU}{2}$.

It should be clear now that the focus of our interest is not so much on the universal \voas{} themselves, but rather on their admissible level simple quotients.  The point is that these simple quotients will have constrained representation theories, due to their additional defining relations, about which we expect to be able to prove classification theorems.  The representation theory of the universal \voas{} is, on the other hand, unconstrained by additional relations so that (almost) every $\alg{g}$-module is allowed.\footnote{One should only exclude modules, such as the adjoint module, that lead to \opes{} with essential singularities.}

To complete the analogy between Virasoro minimal models and fractional level $\AKMA{sl}{2}$ models, we consider the modules of the simple quotient \voas{}.
\begin{defn}
Suppose that $i \in I$ is admissible for a given family $\set{\VOA{V}_i}$ of universal \voas{}.  Then, any module of the simple quotient of $\VOA{V}_i$ is said to be an \emph{admissible module} of $\VOA{V}_i$.
\end{defn}
\noindent Note that the admissible modules of the universal Virasoro \voa{} are precisely the central charge $c = 13 - 6(\tfrac{p}{p'} + \tfrac{p'}{p})$ modules of the Virasoro minimal model $\MinMod{p}{p'}$.  These are the \hw{} Virasoro modules corresponding to the entries in the Kac table.  The admissible modules of the universal $\AKMA{sl}{2}$ \voas{} are the subject of this article.  As noted above, they were first classified in \cite{AdaVer95,DonVer97}.  We remark that they can also be associated with the entries of a table with similar properties to the Virasoro Kac table, see \cite{CreMod13}.

Finally, note that the definition of admissible module given above is not the original definition of Kac and Wakimoto \cite{KacMod88}, who originally defined admissible \hwms{}, for arbitrary affine Kac-Moody algebras, in terms of criteria that guaranteed a character formula, generalising that of Weyl-Kac, and good modular properties.  The \voa{} definition given above is certainly more general and, in our opinion, more fundamental.  However, the two definitions of admissibility coincide for \hw{} $\AKMA{sl}{2}$-modules. A generalised version of this coincidence for higher rank affine Kac-Moody algebras has recently appeared in \cite{AraRat15}.

\subsection{Outline}

We close with a brief outline of the contents of this paper.  \secref{sec:GenHWTh} is a pedagogical introduction to the notion of \rhw{} theory, crucial for studying the $\AdmMod{u}{v}$ models.  The idea actually reduces to a special case of parabolic (also called parahoric or generalised) \hw{} theory, but the connection to \voas{} via Zhu's algebra is so important that we feel it warrants separate consideration.  This is then followed by a detailed discussion of the \rhwms{} of the three Lie algebras used in the remainder of the paper:  The Heisenberg algebra, the (bosonic) $\beta \gamma$ ghost algebra and, of course, $\AKMA{sl}{2}$.

\secref{sec:Wak} opens with a brief review of the Wakimoto free field realisation of $\AKMA{sl}{2}$ and the corresponding Wakimoto modules.  We pay particular attention to the screening fields and operators as a means to motivate the usage of symmetric function theory.  We then discuss the construction of certain \svs{} in \hw{} and \rhw{} $\AKMA{sl}{2}$-modules using Jack symmetric function technology to deduce explicit formulae for these (\thmref{thm:ConstructSVs}).  We believe that these formulae are new:  Even in the standard \hw{} case, the only similar result we are aware of is an old paper of Kato and Yamada \cite{KatMis92} where an integral formula is derived and evaluated in special cases using Schur polynomials.  For \svs{} of \rhwms{}, the only other formulae we know are of the Malikov-Feigin-Fuchs (complex power) form \cite{FeiEqu98}.

One consequence of this \sv{} study is that at admissible levels, Wakimoto's construction yields a free field realisation of the universal \voa{} of $\AKMA{sl}{2}$ rather than of its simple quotient $\AdmMod{u}{v}$.  This is important as it means that the \sv{} of the universal vacuum module is accessible in the Wakimoto realisation, hence it may be exploited to such ends as determining the spectrum of $\AdmMod{u}{v}$-modules.  We first obtain an upper bound on the \hw{} spectrum (\corref{cor:ClassHWMods}) using a surprisingly effortless calculation that combines the form of the vacuum \sv{} with the specialisation formula for Jack symmetric functions.  This bound turns out to be saturated, but to prove this we must address the more involved \rhw{} spectrum.  In this case, a few more symmetric function manipulations allow us to identify an explicit presentation for Zhu's algebra (\thmref{thm:ComputeZhu}) in terms of generators and relations.

It is now easy to classify the simple \rhwms{} of $\AdmMod{u}{v}$, reproducing in an elegant fashion the results of \cite{AdaVer95,DonVer97}.  We emphasise that our proofs are, to the best of our knowledge, independent of those which have appeared in the literature.  We also deduce that \hw{} $\AdmMod{u}{v}$-modules are semisimple, giving a new proof of an old result that (essentially) appeared in \cite{KacMod88}.  In contrast, $\AdmMod{u}{v}$ is shown to admit non-semisimple \rhwms{} with two composition factors and these are characterised using short exact sequences.  We moreover conjecture that these modules, together with their simple analogues, exhaust the indecomposable \rhw{} $\AdmMod{u}{v}$-modules, remarking that proving this would require, among other things, a more detailed knowledge of the submodule structure of \rhw{} $\AKMA{sl}{2}$-modules.  Finally, we prove that the Virasoro zero mode $L_0$ acts semisimply on all \rhw{} $\AdmMod{u}{v}$-modules, independent of our conjecture on the indecomposable spectrum, and discuss briefly how this is consistent with the expectation that, for $v>1$, the $\AdmMod{u}{v}$ model is a \lcft{}.

\appref{sec:SymmPoly} gives a brief, but thorough, introduction to the aspects of symmetric function theory, in particular, those relating to Jack symmetric functions, that are used in the text.  This material is all standard and may be found in \cite{MacSym95}.  We have also included a brief introduction to Zhu's algebra in \appref{sec:zhu}, concentrating on motivating it as an abstract version of the algebra of zero modes acting on the space of ``ground states'' of a \rhwm{}.  This appendix is aimed at physicists, in particular, it uses physics conventions for Fourier expansions, but we hope that it will also prove useful to mathematicians.

\section{Generalising highest weight theory} \label{sec:GenHWTh}

In this section, we consider a generalisation of \hw{} theory that we will qualify as \emph{relaxed}.  Originally introduced for $\AKMA{sl}{2}$ in order to study a correspondence relating $\AKMA{sl}{2}$-modules to those over the $N=2$ superconformal algebra \cite{FeiEqu98}, \rhwms{} have since appeared as necessary constituents of the $\SLG{SL}{2;\RR}$ \WZW{} model \cite{MalStr01}, in admissible level fusion rules \cite{GabFus01,RidFus10}, in relations to logarithmic minimal models \cite{AdaCon05,RidSL210,CreCos13}, in demonstrating the modular invariance of admissible level theories \cite{CreMod12,CreMod13}, and in the full description of bosonic $\beta \gamma$ ghosts \cite{RidFus10,RidBos14}.  Necessity aside, we feel that from some points of view, particularly that of Zhu's algebra \cite{ZhuMod96}, discussed in \appref{sec:zhu}, it is more natural to consider these relaxed modules instead of the standard \hwms{} that one typically encounters in rational \cft{}.

\subsection{Relaxed highest weight theory} \label{sec:Relaxed}

We recall that the formalism of \hw{} theory for a Lie algebra $\alg{g}$ is built from a triangular decomposition
\begin{equation} \label{eq:TriDec}
\alg{g} = \alg{g}_- \oplus \alg{h} \oplus \alg{g}_+.
\end{equation}
This is a vector space direct sum of subalgebras of $\alg{g}$ in which the Cartan subalgebra $\alg{h}$ is abelian and acts semisimply, through the adjoint action, on both $\alg{g}_-$ and $\alg{g}_+$.  In particular, $\comm{\alg{h}}{\alg{g}_{\pm}} \subseteq \alg{g}_{\pm}$.  Moreover, the subalgebras $\alg{g}_-$ and $\alg{g}_+$ are assumed to be antiequivalent in that there exists a (linear) order two antiautomorphism, the adjoint $\dag$, satisfying $\alg{g}_{\pm}^{\dag} = \alg{g}_{\mp}$.  The elements of $\alg{h}$ are supposed to be fixed by $\dag$, meaning that each $x \in \alg{h}$ is self-adjoint \cite{MooLie95}.  This is, however, a little too restrictive; we instead only demand that the adjoint preserves the Cartan subalgebra.  Since $\alg{h}$ is abelian, this implies that each $x \in \alg{h}$ is normal:  $\comm{x}{x^{\dag}} = 0$.  Note that being a linear antiautomorphism means that $\brac{ax}^{\dag} = a x^{\dag}$ and $\comm{x}{y}^{\dag} = \comm{y^{\dag}}{x^{\dag}}$, for all $a \in \CC$ and $x,y \in \alg{g}$.

Given such a triangular decomposition, one defines a \hwv{} in a given $\alg{g}$-module to be a simultaneous eigenvector of the elements of $\alg{h}$ which is annihilated by all the elements of $\alg{g}_+$.  Any module generated by a single \hwv{} is called a \hwm{}.  Conspicuous examples include the Verma modules $\Ver{\lambda} = \UEA \alg{g} \otimes_{\UEA \alg{b}} \CC_{\lambda}$, where $\UEA \alg{g}$ denotes the \uea{} of $\alg{g}$, $\UEA \alg{b}$ that of the Borel subalgebra $\alg{b} = \alg{h} \oplus \alg{g}_+$, $\lambda \in \alg{h}^*$ is the \hw{}, and $\CC_{\lambda}$ is the one-dimensional $\alg{b}$-module upon which $\alg{g}_+$ acts as $0$ and each $x \in \alg{h}$ acts as $\func{\lambda}{x} \in \CC$.  The \hwv{} generating $\Ver{\lambda}$ is $\wun \otimes_{\UEA \alg{b}} 1_{\lambda}$, where $\wun$ denotes the unit of $\UEA \alg{g}$ (and $1_{\lambda}$ that of $\CC_{\lambda} \cong \CC$).

The Lie algebras $\alg{g}$ that one typically encounters in \cft{} have triangular decompositions.  However, they also have more structure in that they are graded by the semisimple action of the Virasoro zero mode $L_0$ so that, if $\alg{g}_n$ denotes the $\func{\ad}{L_0}$-eigenspace of eigenvalue $-n$, then $\comm{\alg{g}_m}{\alg{g}_n} \subseteq \alg{g}_{m+n}$.  Moreover, the Cartan subalgebra $\alg{h}$ is generally chosen to include $L_0$, which we will always assume is self-adjoint, hence it follows that $\alg{g}_m^{\dag} = \alg{g}_{-m}$.  Given this structure, the following definition is natural:
\begin{defn}
Let $\alg{g}$ be a Lie algebra with triangular decomposition \eqref{eq:TriDec} (where the Cartan subalgebra may include elements that are not self-adjoint).  If there exists $L_0 \in \alg{h}$ such that $\alg{g} = \bigoplus_n \alg{g}_n$, where $\alg{g}_n$ is the eigenspace of $\func{\ad}{L_0}$ of eigenvalue $-n$, and $L_0$ is the zero mode of a subalgebra of $\alg{g}$ isomorphic to the Virasoro algebra, then we will say that $\alg{g}$ is \emph{conformally graded}.
\end{defn}
\noindent The most obvious example is the Virasoro algebra itself which is clearly conformally graded with $\alg{g}_n = \vspn \set{L_n}$, for $n \in \ZZ \setminus \set{0}$, $\alg{g}_0 = \vspn \set{L_0, C}$ and $\alg{g}_n = \set{0}$ otherwise.  This shows that the usual, but somewhat confusing, convention that $\comm{L_0}{L_n} = -L_n$ is responsible for $\alg{g}_n$ having $\func{\ad}{L_0}$-eigenvalue $-n$ in the above definition.
\begin{defn}
Given a conformally graded Lie algebra $\alg{g}$, its \emph{relaxed triangular decomposition} is
\begin{equation} \label{eq:RelTriDec}
\alg{g} = \alg{g}_< \oplus \alg{g}_0 \oplus \alg{g}_>,
\end{equation}
where $\alg{g}_< = \bigoplus_{n<0} \alg{g}_n$ and $\alg{g}_> = \bigoplus_{n>0} \alg{g}_n$.  A \emph{\rhwv{}} is then a simultaneous eigenvector of $\alg{h} \subseteq \alg{g}_0$ that is annihilated by $\alg{g}_>$.  A \emph{\rhwm{}} is a module that is generated by a single \rhwv{}.  The \emph{relaxed Borel subalgebra} is $\alg{g}_{\ge} = \alg{g}_0 \oplus \alg{g}_>$ and a \emph{relaxed Verma module} is a $\alg{g}$-module isomorphic to $\RelVer{\Mod{M}} = \UEA \alg{g} \otimes_{\UEA \alg{g}_{\ge}} \Mod{M}$, where $\Mod{M}$ is a simple weight module of $\alg{g}_0$ upon which $\alg{g}_>$ acts as $0$.
\end{defn}
\noindent These definitions have obvious analogues for Lie superalgebras and other more general structures, but we will not need this level of generality in what follows.

In this article, we will only consider triangular decompositions of a conformally graded Lie algebra $\alg{g}$ that satisfy $\alg{g}_< \subseteq \alg{g}_-$ and $\alg{g}_> \subseteq \alg{g}_+$.  Thus, positive modes are always in $\alg{g}_+$ and negative modes are always in $\alg{g}_-$.  When $\alg{g}_0 = \alg{h}$, the relaxed triangular decomposition then coincides with the (unique) triangular decomposition of this type.  We will shortly see examples of \rhwms{} which are not \hwms{} in the usual sense.  First, however, we mention that \rhwms{} may be identified as generalised \hwms{} with respect to the parabolic subalgebra $\alg{g}_{\ge} = \alg{g}_0 \oplus \alg{g}_>$.  We recall that a parabolic subalgebra is any subalgebra that contains a Borel subalgebra (see \cite{MooLie95,HumRep08} for a quick overview of parabolic subalgebras).

It will be occasionally convenient to take this a step further and introduce a category of modules, generalising the well known category $\categ{O}$, that contains the \rhwms{} of $\alg{g}$.  This is essentially (a variant of) the parabolic category $\categ{O}$ discussed, for example, in \cite{HumRep08}.  The explicit details of this category are not essential for understanding the results to follow, but we shall devote a few paragraphs to explaining their physical motivation.
\begin{defn}
Given a conformally graded Lie algebra $\alg{g} = \alg{g}_< \oplus \alg{g}_0 \oplus \alg{g}_>$, define the \emph{relaxed category $\categ{R}$} to consist of the $\alg{g}$-modules $\Mod{M}$ satisfying the following axioms:
\begin{itemize}
\item $\Mod{M}$ is finitely generated.
\item $\Mod{M}$ is a weight module (the action of the Cartan subalgebra $\alg{h}$ is semisimple).
\item The action of $\alg{g}_>$ is locally nilpotent:  For each $v \in \Mod{M}$, the space $\UEA \alg{g}_> \cdot v$ is finite-dimensional.
\end{itemize}
The morphisms are the $\alg{g}$-module homomorphisms between these modules, as usual.
\end{defn}
\noindent All \hw{} and \rhwms{} belong to category $\categ{R}$.  Moreover, if $\alg{g}$ has finite-dimensional root spaces, then it follows that each module in $\Mod{M}$ will have finite-dimensional weight spaces.  Another important consequence of these axioms is that every (non-zero) module in category $\categ{R}$ possesses a \rhwv{}, hence that every simple category $\categ{R}$ module is a \rhwm{}.

One can, and should, ask whether the mathematical axioms that we impose on category $\categ{R}$ will end up excluding modules relevant for applications.  This is an important question and the answer is that they do exclude relevant modules, but in a well-controlled manner.  Our motivation for introducing this category is that we want to classify the modules of the admissible level $k$ \voa{} $\AdmMod{u}{v}$ by identifying these modules as $\AKMA{sl}{2}$-modules.  For the physical application of investigating the corresponding \cfts{}, we must insist that the category of $\AdmMod{u}{v}$-modules be closed under the conjugation operation of $\AKMA{sl}{2}$ (see \secref{sec:RelEx}) and fusion.  Moreover, we want the characters to behave well under modular transformations so that one can identify modular invariant partition functions and compute (Grothendieck) fusion rules from a Verlinde-type formula.

Category $\categ{O}$ is not sufficient for these purposes, in particular, the conjugate of an $\AdmMod{u}{v}$-module from category $\categ{O}$ need not lie in category $\categ{O}$.  Relaxing to category $\categ{R}$ alleviates this problem and has recently been shown \cite{CreMod13} to lead to characters with excellent modular behaviour, provided that one extends the category again to take into account twists by the so-called spectral flow automorphisms.  The upshot is that these spectrally-flowed modules are not in category $\categ{R}$, but the twisting is very well understood, justifying the above statement that the exclusion of these physically relevant modules is under control.

Axiomatically, accounting for the spectrally-flowed modules would require weakening the local nilpotency axiom above.  However, this axiom has the advantage that category $\categ{R}$ provides a very natural setting for the important, and very useful, technology of Zhu's algebra, discussed in \appref{sec:zhu}.  We restrict to weight modules because the fusion coproduct formulae \cite{GabFus94} for $\AKMA{sl}{2}$ show that the fusion product of two weight modules will again be weight.  Similarly, the conjugate of a weight module is weight and omitting non-weight modules does not restrict the characters in any way.  Moreover, being weight does not preclude the Virasoro zero mode $L_0$ from acting non-semisimply as is required in \lcfts{}.

To summarise, the relaxed category $\categ{R}$ is a rich source of modules for affine Kac-Moody (super)algebras that appears to be even more relevant to \cft{} than the much more familiar category $\categ{O}$.

\subsection{Examples} \label{sec:RelEx}

In this paper, we will use the Wakimoto free field realisation to study the \rhwms{} of $\AKMA{sl}{2}$ and determine which of these are modules for the admissible level \voas{} $\AdmMod{u}{v}$, where $u \in \ZZ_{\ge 2}$ and $v \in \ZZ_{\ge 1}$ are coprime and the level $k$ is determined by $\frac{u}{v} = t = k+2$.  We will therefore need to investigate the (relaxed) \hw{} theory of the Heisenberg algebra $\HA$, the $c=2$ bosonic \(\beta\gamma\) ghost system $\GA$ and $\AKMA{sl}{2}$ itself.  This investigation constitutes the rest of the section.

\subsection*{The Heisenberg algebra $\HA$}

We will use the same notations and conventions for the Heisenberg algebra as in \cite{RidJac14}.  The free boson \voa{} is generated by a single bosonic field \(a(z)\), defined by the \ope{}
\begin{equation} \label{eq:HeisOPE}
  a(z)a(w)\sim\frac{\wun}{(z-w)^2}.
\end{equation}
With the standard Fourier expansion, $a(z)=\sum_{n\in\mathbb{Z}}a_n z^{-n-1}$, the \ope{} implies the following commutation relations:
\begin{equation} \label{eq:BosComm}
  \comm{a_m}{a_n}=m\delta_{m+n,0}\wun, \qquad m,n\in\mathbb{Z}.
\end{equation}
The Heisenberg algebra $\HA$ is then the infinite-dimensional Lie algebra spanned by the $a_n$ and $\wun$, the latter being identified with the unit of the \uea{} of \(\HA\), as usual.  We will assume that \(\wun\) acts as the identity operator on any \(\HA\)--module.  This is only a minor restriction since a simple rescaling of the generators lets this operator act as multiplication by any non-zero number.

As is well known, the free boson \voa{} admits a one-parameter family of conformal structures.  We will write the energy-momentum tensor and central charge in the form
\begin{equation} \label{eq:HeisT}
    T^{\textup{bos.}}(z)=\frac{1}{2}\normord{a(z)a(z)}-\frac{1}{\alpha}\partial a(z),\qquad c^{\textup{bos.}}=1-\frac{12}{\alpha^2},
\end{equation}
where $\alpha$ parametrises the conformal structure.  We note that $\alpha \to \infty$ reproduces the standard free boson central charge $c^{\textup{bos.}} = 1$.  In Wakimoto's construction, this would correspond to $k \to \infty$, so it is permissible to ignore this case.  It is worth recalling that $a(z)$ is not a Virasoro primary for $\alpha$ finite; instead we have
\begin{equation}
T^{\textup{bos.}}(z)a(w)\sim\frac{2/\alpha\:\wun}{(z-w)^3}+\frac{a(w)}{(z-w)^2}+\frac{\pd a(w)}{z-w}.
\end{equation}

The Fourier expansion $T^{\textup{bos.}}(z) = \sum_{n \in \ZZ} L_n^{\textup{bos.}} z^{-n-2}$ defines the Virasoro modes and it is easy to check that the Lie algebra $\alg{g}$ spanned by the $a_n$, $L_n^{\text{bos.}}$ and $\wun$ is conformally graded.  It is likewise easy to check that (for finite $\alpha$) the only adjoint on the Heisenberg algebra consistent with the standard Virasoro adjoint $(L_n^{\textup{bos.}})^{\dag} = L_{-n}^{\textup{bos.}}$ is
\begin{equation} \label{eq:HeisAdj}
a_n^{\dag} = -a_{-n} - \frac{2}{\alpha} \delta_{n,0} \wun.
\end{equation}
Since $\alg{g}_0 = \vspn \set{a_0, L_0, \wun}$ is abelian, it coincides with the Cartan subalgebra $\alg{h}$, hence \rhw{} theory reduces to ordinary \hw{} theory for the Heisenberg algebra.  We remark that this is one example where we cannot insist that the Cartan subalgebra consist of self-adjoint elements.

The \hw{} theory of the Heisenberg Lie algebra is well known.  The Verma modules are known as Fock spaces and are parametrised by the $a_0$-eigenvalue $p$ of the \hwv{}.  They are simple for all \(p\in\mathbb{C}\). We will denote the Fock spaces by \(\FF{p}\) and their corresponding \hwvs{} by \(\hv{p}\).  In accordance with \eqref{eq:HeisAdj}, the module conjugate to the Fock space $\FF{p}$ is $\FF{-p-2/\alpha}$.

\subsection*{The ghost algebra $\GA$}

For the $\beta \gamma$ ghost system, we follow the notations and conventions of \cite{RidBos14}.  The ghost \voa{} is generated by two bosonic fields, \(\beta(z)\) and \(\gamma(z)\), whose \opes{} are
\begin{equation} \label{eq:GhOPE}
  \beta(z)\beta(w)\sim 0, \qquad 
  \gamma(z)\beta(w)\sim\frac{\wun}{z-w}, \qquad 
  \gamma(z)\gamma(w)\sim 0.
\end{equation}
From these, one constructs a Heisenberg field $J(z)$ and an energy-momentum tensor $T^{\textup{gh.}}(z)$ by
\begin{equation}
J(z) = \normord{\beta(z)\gamma(z)}, \qquad 
T^{\textup{gh.}}(z) = -\normord{\beta(z)\pd\gamma(z)}.
\end{equation}
These give $\beta(z)$ and $\gamma(z)$ Heisenberg weights $+1$ and $-1$, and conformal weights $1$ and $0$, respectively.  We remark that $J(z)$ is not normalised as in \eqref{eq:HeisOPE}, nor is it primary with respect to $T^{\textup{gh.}}(z)$:
\begin{equation}
J(z)J(w)\sim\frac{-\wun}{(z-w)^2}, \qquad 
T^{\textup{gh.}}(z)J(w)\sim\frac{-\wun}{(z-w)^3}+\frac{J(w)}{(z-w)^2}+\frac{\pd J(w)}{z-w}.
\end{equation}
As with the free boson, there is actually a one-parameter family of conformal structures; $T^{\textup{gh.}}(z)$ has been chosen so that the ghost fields $\beta(z)$ and $\gamma(z)$ have the required conformal weights.  This choice also fixes the central charge of the ghost system to be $c^{\textup{gh.}}=2$.

The Fourier expansions $\beta(z)=\sum_{n\in\mathbb{Z}}\beta_n z^{-n-1}$ and $\gamma(z)=\sum_{n\in\mathbb{Z}}\gamma_nz^{-n}$ now yield the commutation relations
\begin{equation} \label{eq:BosGhComm}
  \comm{\beta_m}{\beta_n}=0,\quad
  \comm{\gamma_m}{\beta_n}=\delta_{m+n,0}\wun,\quad
  \comm{\gamma_m}{\gamma_n}=0;\qquad
  m,n\in\mathbb{Z}.
\end{equation}
The infinite-dimensional Lie algebra $\GA$, spanned by the $\beta_n$, $\gamma_n$ and $\wun$, is called the ghost Lie algebra.  Again, \(\wun\) is identified with the unit of $\UEA \GA$ and we assume that it acts as the identity on all \(\GA\)-modules.  As with the Heisenberg algebra, we will extend this algebra by the modes $J_n, L_n^{\textup{gh.}} \in \UEA \GA$, defined by $J(z) = \sum_{n \in \ZZ} J_n z^{-n-1}$ and $T^{\textup{gh.}}(z) = \sum_{n \in \ZZ} L_n^{\textup{gh.}} z^{-n-2}$.  The Lie algebra $\alg{g}$ spanned by the $\beta_n$, $\gamma_n$, $J_n$, $L_n^{\textup{gh.}}$ and $\wun$ is then conformally graded with relations including
\begin{subequations}
\begin{gather}
\comm{J_m}{\beta_n} = \beta_{m+n}, \qquad 
\comm{J_m}{\gamma_n} = -\gamma_{m+n}, \qquad 
\comm{L_m^{\textup{gh.}}}{\beta_n} = -n \beta_{m+n}, \qquad 
\comm{L_m^{\textup{gh.}}}{\gamma_n} = -(m+n) \gamma_{m+n}, \\
\comm{J_m}{J_n} = -m \delta_{m+n,0} \wun, \qquad 
\comm{L_m^{\textup{gh.}}}{J_n} = - n J_{m+n} - \frac{1}{2} m(m+1) \delta_{m+n,0} \wun.
\end{gather}
\end{subequations}
The Cartan subalgebra $\alg{h} = \vspn \tset{J_0, L_0^{\textup{gh.}}, \wun}$ is a proper subalgebra of $\alg{g}_0 = \vspn \tset{\beta_0, \gamma_0, J_0, L_0^{\textup{gh.}}, \wun}$, so the relaxed and ordinary \hw{} theories of the $\beta \gamma$ ghost system do not coincide.  In this case, the ghost adjoint
\begin{equation} \label{eq:GhAdj}
\beta_n^{\dag} = \gamma_{-n}
\end{equation}
implies that all the elements of $\alg{h}$ are self-adjoint.

We start with the ordinary \hw{} theory corresponding to the triangular decomposition in which $\beta_0 \in \alg{g}_+$ annihilates the \hwv{} $\gvac$ and $\gamma_0 \in \alg{g}_-$ need not.  It turns out that this yields a unique Verma module because these conditions imply that $J_0 \gvac = L_0^{\textup{gh.}} \gvac = 0$ (and $\wun$ acts as the identity, as always).  This module is simple; in fact, we may take it to be the ghost vacuum module $\GhostVac$ because these conditions also imply that $L_{-1}^{\textup{gh.}} \gvac = 0$.  Of course, one can also take the triangular decomposition in which $\gamma_0$ annihilates the \hwv{} and $\beta_0$ need not.\footnote{There are other triangular decompositions, but they will not concern us here.  Indeed, they may be obtained from those already mentioned by twisting with a so-called spectral flow automorphism, see \cite{RidBos14}.}  The resulting Verma module is the simple module $\ghconjmod{\GhostVac}$ conjugate to $\GhostVac$, obtained by twisting the action of $\GA$ by $\ghconjaut$, the (order $4$) conjugation automorphism that sends $\beta_n$ to $\gamma_n$ and $\gamma_n$ to $-\beta_n$.

The \rhw{} theory of the ghost system is significantly more interesting as we no longer require that $\beta_0$ (or $\gamma_0$) annihilates the \hwv{}.  In fact, we may induce from a fairly arbitrary simple weight module of $\alg{g}_0$.  Given that $J_0$ and $L_0^{\textup{gh.}}$, as abstract elements of $\alg{g}_0$, are going to be identified with elements of $\UEA \GA$ and that the $\alg{g}_0$-module weight vectors are going to be identified with the \rhwvs{} of the induced module, we may restrict to $\alg{g}_0$-modules on which the action of $J_0$ is identified with that of $\gamma_0 \beta_0$ and the action of $L_0^{\textup{gh.}}$ is always $0$.
\begin{cprop}[\protect{\cite[Prop.~1]{RidBos14}}] \label{prop:FinGhReps}
A simple weight module over $\alg{g}_0$ upon which $J_0 = \gamma_0 \beta_0$
and $L_0^{\textup{gh.}} = 0$ is isomorphic to one of the following:
\begin{itemize}
\item The module $\FinGhostVac$ generated by a vector \(\gv{}\) which is annihilated by \(\beta_0\) and thus also \(J_0\).  This module has a basis of weight vectors $\gv{j}$, $j \in \ZZ_{\le 0}$, where $J_0 \gv{j} = j \gv{j}$.
\item The module $\finghconjmod{\FinGhostVac}$, conjugate to $\FinGhostVac$, generated by a vector \(\gv{}\) which is annihilated by \(\gamma_0\) and thus \(J_0\gv{}=\gv{}\).  This module has a basis of weight vectors $\gv{j}$, $j \in \ZZ_{\ge 1}$, where $J_0 \gv{j} = j \gv{j}$.
\item The modules $\FinGhostRel{q}$, where $q\in\CC\setminus\ZZ$, each of which is generated by a vector \(\gv{}\) satisfying \(J_0\gv{}=q\gv{}\); it follows that no non-zero vector is annihilated by $\beta_0$ or $\gamma_0$.  The eigenvalues of $J_0 = \gamma_0 \beta_0$ all lie in $q + \ZZ$ and these modules have a basis of weight vectors $\gv{j}$, $j \in q + \ZZ$, satisfying $J_0 \gv{j} = j \gv{j}$.  The modules $\FinGhostRel{q}$ and $\FinGhostRel{q+1}$ are isomorphic.
\end{itemize}
\end{cprop}
Additionally, one can consider the indecomposable $\alg{g}_0$-modules $\FinGhostRel{0}^+$ and $\FinGhostRel{0}^-$ that likewise have a basis of weight vectors $\gv{j}$, $j \in \ZZ$, satisfying $J_0 \gv{j} = j \gv{j}$, but they are not simple.  They are determined (up to isomorphism) by the following non-split short exact sequences:
\begin{equation}
\dses{\FinGhostVac}{}{\FinGhostRel{0}^+}{}{\finghconjmod{\FinGhostVac}}, \qquad 
\dses{\finghconjmod{\FinGhostVac}}{}{\FinGhostRel{0}^-}{}{\FinGhostVac}.
\end{equation}

Inducing the $\alg{g}_0$-modules $\FinGhostVac$, $\finghconjmod{\FinGhostVac}$ and $\FinGhostRel{q}$ ($q \notin \ZZ$) therefore results in relaxed Verma modules for $\GA$.  The first two give the simple ghost vacuum module $\GhostVac$ and its conjugate $\ghconjmod{\GhostVac}$, respectively; the last gives new modules $\GhostRel{q}$ which are also simple and satisfy $\GhostRel{q} \cong \GhostRel{q+1}$.  We may similarly induce the non-simple $\alg{g}_0$-modules $\FinGhostRel{0}^+$ and $\FinGhostRel{0}^-$ to obtain non-simple $\GA$-modules $\FinGhostRel{0}^+$ and $\FinGhostRel{0}^-$ which are likewise determined (up to isomorphism) by the following non-split short exact sequences:
\begin{equation}
\dses{\GhostVac}{}{\GhostRel{0}^+}{}{\ghconjmod{\GhostVac}}, \qquad 
\dses{\ghconjmod{\GhostVac}}{}{\GhostRel{0}^-}{}{\GhostVac}.
\end{equation}

\subsection*{The affine Kac-Moody algebra $\AKMA{sl}{2}$}

We consider the universal \voa{} of non-critical level $k \neq -2$ generated by three fields $e(z)$, $h(z)$ and $f(z)$ satisfying the \opes{}
\begin{equation} \label{eq:SL2OPE}
\begin{aligned}
\func{h}{z} \func{e}{w} &\sim \frac{+2 \func{e}{w}}{z-w} \vphantom{\frac{2k}{\brac{z-w}^2}}, \\
\func{h}{z} \func{f}{w} &\sim \frac{-2 \func{f}{w}}{z-w} \vphantom{\frac{-k}{\brac{z-w}^2}},
\end{aligned}
\qquad
\begin{aligned}
\func{h}{z} \func{h}{w} &\sim \frac{2k\:\wun}{\brac{z-w}^2}, \\
\func{e}{z} \func{f}{w} &\sim \frac{-k\:\wun}{\brac{z-w}^2} - \frac{\func{h}{w}}{z-w},
\end{aligned}
\qquad
\begin{aligned}
\func{e}{z} \func{e}{w} &\sim 0 \vphantom{\frac{2k}{\brac{z-w}^2}}, \\
\func{f}{z} \func{f}{w} &\sim 0 \vphantom{\frac{-k}{\brac{z-w}^2}},
\end{aligned}
\end{equation}
and no other (independent) relations. We denote this \voa{} by \(\UVOA{k}\).  The maximal proper ideal of \(\UVOA{k}\) is non-trivial if and only if \(k\) is admissible.  Moreover, this ideal is generated by a single primary field (singular vector) \cite{KacStr79}; we do not set this primary field to zero.  Note that we have chosen the $\SLA{sl}{2}$ basis $\set{e,h,f}$ to be consistent with previous work, \cite{RidSL208} in particular, where it was necessary to tailor the basis to the $\SLA{sl}{2;\RR}$ adjoint rather than the (more traditional) $\SLA{su}{2}$ adjoint.  This is reflected in the signs appearing in the formulae for $e(z)f(w)$ above.

With the usual Fourier expansions $g(z) = \sum_{n \in \ZZ} g_n z^{-n-1}$,
where $g=e,h,f$, the commutation relations are
\begin{equation} \label{eq:CommSL2}
\begin{aligned}
\comm{h_m}{e_n} &= +2 e_{m+n}, \\
\comm{h_m}{f_n} &= -2 f_{m+n},
\end{aligned}
\quad
\begin{aligned}
\comm{h_m}{h_n} &= 2m \delta_{m+n,0} k \: \wun, \\
\comm{e_m}{f_n} &= -h_{m+n} - m \delta_{m+n,0} k \: \wun,
\end{aligned}
\quad
\begin{aligned}
\comm{e_m}{e_n} &= 0, \\
\comm{f_m}{f_n} &= 0,
\end{aligned}
\qquad m,n \in \ZZ
\end{equation}
and these make $\vspn \set{e_n, h_n, f_n, \wun}$ into a Lie algebra which we denote by $\AKMA{sl}{2}$.  Once again, we assume that the unit $\wun \in \UEA \tbrac{\AKMA{sl}{2}}$ acts as the identity on each $\AKMA{sl}{2}$-module.

The standard conformal structure of \(\UVOA{k}\) 
is uniquely determined by requiring that $e(z)$, $h(z)$ and $f(z)$ are Virasoro primaries of conformal weight $1$ (this structure exists for all $k \neq -2$).  The Sugawara construction then gives the explicit form of the energy-momentum tensor as
\begin{equation} \label{eq:DefT}
\func{T}{z} = \frac{1}{2t} \brac{\frac{1}{2} \normord{\func{h}{z} \func{h}{z}} - \normord{\func{e}{z} \func{f}{z}} - \normord{\func{f}{z} \func{e}{z}}},
\end{equation}
where $t = k+2$.  With $\tfunc{T}{z} = \sum_{n \in \ZZ} L_n z^{-n-2}$, one finds that the modes $L_n$ generate a copy of the Virasoro algebra with central charge $c = 3 - 6/t$.  The Lie algebra $\alg{g}$ spanned by the $e_n$, $h_n$, $f_n$, $L_n$ and $\wun$ is then conformally graded with Cartan subalgebra $\alg{h} = \vspn \set{h_0, L_0, \wun}$.  We have chosen the $\SLA{sl}{2}$ basis so that the $\SLA{sl}{2;\RR}$ adjoint becomes
\begin{equation}
e_n^{\dag} = f_{-n}, \qquad h_n^{\dag} = h_{-n}.
\end{equation}
Again, the Cartan subalgebra consists of self-adjoint elements.

The \hw{} theory of $\AKMA{sl}{2}$ is well known.  The standard triangular decomposition splits the zero modes so that $e_0 \in \alg{g}_+$ annihilates \hwvs{} but $f_0 \in \alg{g}_-$ need not.  Then, the Verma modules $\Ver{\lambda}$ are parametrised by the $\SLA{sl}{2}$-weight ($h_0$-eigenvalue) $\lambda \in \CC$ of the \hwv{} because it follows from \eqref{eq:DefT} that its conformal weight is then given by
\begin{equation} \label{eq:ConfDimHWV}
\Delta_{\lambda} = \frac{\lambda \brac{\lambda + 2}}{4t}. 
\end{equation}
These Verma modules need not be simple.  The quotient module $\Ver{0} / \Ver{-2}$, where the submodule $\Ver{-2}$ is generated by acting with $f_0$ on the (generating) \hwv{} of $\Ver{0}$, is the vacuum module; it carries the structure of the universal \voa{} \(\UVOA{k}\) defined by \eqref{eq:SL2OPE}.  As with the ghosts, one can also consider the triangular decomposition in which $f_0$ annihilates \hwvs{} but $e_0$ need not.\footnote{And as with the ghosts, there are again other triangular decompositions that will not concern us here, being related to those discussed here by spectral flow automorphisms, see \cite{RidSL208,CreMod13}.}  The Verma modules with respect to this decomposition are then the conjugates $\conjmod{\Ver{\lambda}}$ of the Verma modules $\Ver{\lambda}$, where the (order $2$) conjugation automorphism $\conjaut$ sends $e_n$ to $-f_n$ and $h_n$ to $-h_n$.

Because $\alg{g}_0 = \vspn \set{e_0, h_0, f_0, L_0, \wun}$ is non-abelian, it strictly contains $\alg{h}$, so the \rhw{} theory of $\AKMA{sl}{2}$ is strictly more general.  We note that inducing from a $\alg{g}_0$-module reduces to choosing an $\SLA{sl}{2}$-module because we require that $4t L_0$ acts as $h_0^2 - 2 e_0 f_0 - 2 f_0 e_0$ (this is how $4t L_0$ acts on \rhwvs{}).  The analogue of \propref{prop:FinGhReps} is then the classification of weight modules for $\SLA{sl}{2}$ (see \cite{MazLec10}, for example):
\begin{prop} \label{prop:FinSL2WtMods}
The simple weight modules of $\SLA{sl}{2}$ are exhausted by the following:
\begin{itemize}
\item The $\brac{\lambda + 1}$-dimensional modules $\FinIrr{\lambda}$, with $\lambda \in \ZZ_{\ge 0}$.  The module $\FinIrr{\lambda}$ has a basis of weight vectors $\av{m}$, where $m = \lambda, \lambda -2 , \ldots, -\lambda$ and $h_0 \av{m} = m \av{m}$.  It is both highest and lowest weight.
\item The infinite-dimensional \hwms{} $\FinDisc{\lambda}$, with $\lambda \in\mathbb{C}\setminus \ZZ_{\ge 0}$.  The module $\FinDisc{\lambda}$ has a basis of weight vectors $\av{m}$, where $m = \lambda, \lambda -2, \lambda - 4, \ldots$ and $h_0 \av{m} = m \av{m}$.
\item The infinite-dimensional lowest weight modules $\finconjmod{\FinDisc{-\lambda}}$, with $\lambda \in\mathbb{C}\setminus \ZZ_{\le 0}$.  Here, $\finconjaut$ is the Weyl reflection of $\SLA{sl}{2}$ that sends $e_0$ to $-f_0$ and $h_0$ to $-h_0$. The module $\finconjmod{\FinDisc{-\lambda}}$ has a basis of weight vectors $\av{m}$, where $m = \lambda, \lambda + 2, \lambda + 4, \ldots$ and $h_0 \av{m} = m \av{m}$.
\item The infinite-dimensional weight modules $\FinRel{\lambda; \Delta}$, with $\lambda, \Delta \in \CC$ satisfying $4t \Delta \neq \mu \brac{\mu+2}$ for any $\mu \in \lambda + 2 \ZZ$.  The module $\FinRel{\lambda; \Delta}$ has a basis of weight vectors $\av{\mu}$, with $\mu \in \lambda + 2 \ZZ$, satisfying $h_0 \av{\mu} = \mu \av{\mu}$.  It is neither highest nor lowest weight.  Moreover, there are isomorphisms $\FinRel{\lambda; \Delta} \cong \FinRel{\lambda + 2; \Delta}$.
\end{itemize}
\end{prop}

\noindent As with \propref{prop:FinGhReps}, there are non-simple analogues of the $\FinRel{\lambda; \Delta}$ when $4t \Delta = \mu \brac{\mu+2}$ for some $\mu \in \lambda + 2 \ZZ$.  The structures of these indecomposables depend upon precisely how many $\mu \in \lambda + 2 \ZZ$ satisfy this constraint \cite{MazLec10} and we shall defer their consideration until they are needed (\secref{sec:HWRelAdmMod}).

Inducing each of the simple $\SLA{sl}{2}$-modules now yields relaxed Verma modules for $\AKMA{sl}{2}$.  More precisely:
\begin{itemize}
\item Inducing $\FinIrr{\lambda}$, where $\lambda \in \ZZ_{\ge 0}$, results in the \hwm{} $\Ver{\lambda} / \Ver{-\lambda - 2}$, where the \hwvs{} of $\Ver{\lambda}$ and $\Ver{-\lambda - 2}$ are related by $\av{-\lambda - 2} = f_0^{\lambda + 1} \av{\lambda}$.  This induced module need not be simple; its simple quotient will be denoted by $\Irr{\lambda}$.  These simple modules are self-conjugate:  $\conjmod{\Irr{\lambda}} = \Irr{\lambda}$.
\item Inducing $\FinDisc{\lambda}$, where $\lambda \in \CC \setminus \ZZ_{\ge 0}$, results in the Verma module $\Ver{\lambda}$; the simple quotient will be denoted by $\Disc{\lambda}$.
\item Inducing $\finconjmod{\FinDisc{-\lambda}}$, where $\lambda \in \CC \setminus \ZZ_{\le 0}$, results in $\conjmod{\Ver{-\lambda}}$; the simple quotient is $\conjmod{\Disc{-\lambda}}$.
\item Inducing $\FinRel{\lambda; \Delta}$, with $\lambda, \Delta \in \CC$ satisfying $4t \Delta \neq \mu \brac{\mu+2}$ for any $\mu \in \lambda + 2 \ZZ$, results in a new relaxed Verma module that we shall denote by $\RelVer{\lambda; \Delta}$; its simple quotient will be denoted by $\Rel{\lambda; \Delta}$.  As above, there are isomorphisms $\RelVer{\lambda; \Delta} \cong \RelVer{\lambda + 2; \Delta}$ and $\Rel{\lambda; \Delta} \cong \Rel{\lambda + 2; \Delta}$.  Finally, the module conjugate to $\Rel{\lambda; \Delta}$ is $\Rel{-\lambda; \Delta}$.
\end{itemize}
We emphasise that whereas the simple \hwms{} (and their conjugates) are characterised by a single parameter, the \hw{}, the $\Rel{\lambda; \Delta}$ require two parameters in general.  The three classes of simple $\AKMA{sl}{2}$-modules $\Irr{\lambda}$, $\Disc{\lambda}$ and $\conjmod{\Disc{-\lambda}}$, and $\Rel{\lambda; \Delta}$ are distinguished by their \rhwvs{}:  $\Irr{\lambda}$ has finitely many ($\lambda + 1$ in fact); $\Disc{\lambda}$ and $\conjmod{\Disc{-\lambda}}$ have infinitely many, but their $\SLA{sl}{2}$-weights are bounded above and below, respectively; $\Rel{\lambda; \Delta}$ has infinitely many with no bound on the $\SLA{sl}{2}$-weights.

\section{The Wakimoto Free Field Realisation} \label{sec:Wak}

A free field realisation of $\AKMA{sl}{2}$ for any level was constructed by Wakimoto in \cite{WakFoc86}.  It shows that fields $\tfunc{e}{z}$, $\tfunc{h}{z}$ and $\tfunc{f}{z}$, satisfying the \opes{} \eqref{eq:SL2OPE}, may be constructed in terms of a free boson and a pair of bosonic ghosts.  We review this and the screening operator formalism of free field theories, concluding by deriving explicit formulae, in terms of Jack symmetric polynomials, for certain (relaxed) singular vectors in (relaxed) Wakimoto modules over $\AKMA{sl}{2}$.  A corollary of this analysis is that for admissible levels, Wakimoto's construction describes the universal \voa{} $\UVOA{k}$ of $\AKMA{sl}{2}$, rather than its simple quotient $\AdmMod{u}{v}$.

\subsection{The Wakimoto construction}

There is a one-parameter family of realisations of \(\AKMA{sl}{2}\) in terms of tensor products of the free boson and ghost fields.   Given the \opes{} \eqref{eq:HeisOPE} and \eqref{eq:GhOPE}, it is straightforward to verify that defining (we omit the tensor products for notational simplicity)
\begin{equation}\label{eq:ffrealisation}
\begin{gathered}
e(z)=\beta(z),\qquad
h(z)=2\normord{\beta(z)\gamma(z)}+\alpha\:a(z), \\
f(z)=\normord{\beta(z)\gamma(z)\gamma(z)}+\alpha\:a(z)\gamma(z)+\Bigl(\frac{\alpha^2}{2}-2\Bigr)\pd\gamma(z)
\end{gathered}
\end{equation}
reproduces the $\AKMA{sl}{2}$ \opes{} \eqref{eq:SL2OPE} with the level $k = t-2$ being related to the parameter \(\alpha\) by $\alpha^2 = 2t$.  Moreover, the $\AKMA{sl}{2}$ energy-momentum tensor \eqref{eq:DefT} and central charge then decompose as
\begin{equation} \label{eq:WakT}
  \begin{split}
    T(z)&=T^{\text{bos.}}(z) + T^{\text{gh.}}(z) = \frac{1}{2}\normord{a(z)a(z)}-\frac{1}{\alpha}\pd
    a(z)-\normord{\beta(z)\pd\gamma(z)},\\
    c&=c^{\text{bos.}} + c^{\text{gh.}} = 1-\frac{12}{\alpha^2} + 2=3-\frac{6}{t},
  \end{split}
\end{equation}
identifying $\alpha$ with the deformation parameter in the free boson conformal structure \eqref{eq:HeisT}. 

\begin{defn}
We define the \emph{Wakimoto \voa{}} to be the tensor product of the Heisenberg and ghost \voas{}, equipped with the conformal structure given in \eqref{eq:WakT}.
\end{defn}

\noindent Recall that when $k$ is admissible, the universal $\AKMA{sl}{2}$ \voa{} is not simple, but has a unique maximal ideal that is generated by a single primary field (singular vector).  One of the results of this section (\corref{cor:WakVacIsUniversal}) is that this field is non-zero in the Wakimoto \voa{}.

The Wakimoto free field realisation endows the tensor product of a Heisenberg Fock space and a ghost module with the structure of an \(\AKMA{sl}{2}\)-module, by restriction.  We refer to such modules as Wakimoto modules, distinguishing at least four types:
\begin{itemize}
\item The \hw{} Wakimoto modules $\Wak{p} = \FF{p} \otimes \GhostVac$, recalling that $\GhostVac$ denotes the ghost vacuum module.
\item The conjugate \hw{} Wakimoto modules $\ghconjmod{\Wak{p}} = \FF{p} \otimes \ghconjmod{\GhostVac}$, obtained by conjugating the ghost module.
\item The \rhw{} Wakimoto modules $\WakRel{p; q} = \FF{p} \otimes \GhostRel{q}$, for $q \notin \ZZ$.
\item The \rhw{} Wakimoto modules $\WakRel{p; 0}^{\pm} = \FF{p} \otimes \GhostRel{0}^{\pm}$.
\end{itemize}
The structure of the \hw{} Wakimoto modules $\Wak{p}$ was determined in \cite{BerFoc90}.  To the best of our knowledge, the relaxed modules $\WakRel{p;q}$ have not previously been considered in the literature.

Consider now the tensor product of a Heisenberg \hwv{} $\hv{p}$ and the ghost vacuum $\gvac$, denoting it by $\wv{p} = \hv{p} \otimes \gvac$ for convenience.  The free field realisations \eqref{eq:ffrealisation} of the $\AKMA{sl}{2}$ fields imply that $\wv{p}$ is an $\AKMA{sl}{2}$ \hwv{} of $\SLA{sl}{2}$- and conformal weight
\begin{equation}
\lambda_p = \alpha p, \qquad \Delta_p = \frac{1}{2} p \brac{p + \frac{2}{\alpha}} = \frac{\lambda_p \brac{\lambda_p + 2}}{4t},
\end{equation}
where we recall that $\alpha^2 = 2t$.  Similarly, the tensor product $\wv{p;q} = \hv{p} \otimes \gv{q}$ of $\hv{p}$ with a \rhwv{} $\gv{q}$ for the ghosts is an $\AKMA{sl}{2}$ \rhwv{} of $\SLA{sl}{2}$- and conformal weight
\begin{equation} \label{eq:RelaxedWts}
\lambda_{p;q} = \alpha p + 2q, \qquad \Delta_p = \frac{1}{2} p \brac{p + \frac{2}{\alpha}} = \frac{\brac{\lambda_{p;q} - 2q} \brac{\lambda_{p;q} - 2q + 2}}{4t}.
\end{equation}
This shows that all the simple $\AKMA{sl}{2}$-modules, \hw{} and relaxed, may be realised as subquotients of Wakimoto modules.

\subsection{Screening fields and operators}

We begin by recalling the construction of vertex operators for the free boson. For this, one extends the Heisenberg algebra by introducing a generator \(\hat a\) satisfying
\begin{equation}
  \comm{a_m}{\hat a}=\delta_{m,0}\wun.
\end{equation}
The vertex operators $\vop{p}(z)$, parametrised by $p \in \CC$, are then defined by
\begin{equation}\label{eq:vertop}
  \vop{p}(z)=\ee^{p\hat a}z^{p a_0}
  \prod_{m\ge 1}\exp\brac{p\frac{\alpha_{-m}}{m}z^m}
  \exp\brac{-p\frac{\alpha_{m}}{m}z^{-m}}.
\end{equation}
A standard computation shows that the vertex operators are free boson primaries of Heisenberg weight $p$ and conformal weight $\frac{1}{2} p \brac{p + \frac{2}{\alpha}}$, by virtue of the \opes{}
\begin{equation}
  a(z)\vop{p}(w)\sim\frac{p \vop{p}(w)}{z-w}, \qquad
  T^{\text{bos.}}(z)\vop{p}(w)\sim \frac{\frac{1}{2} p \big( p + \frac{2}{\alpha} \bigr) \vop{p}(w)}{(z-w)^2} + \frac{\pd\vop{p}(w)}{z-w}.
\end{equation}
For later use, we record that the composition of \(k\) vertex operators is given by
\begin{equation} \label{eq:CompVerOps}
  \vop{p_1}(z_1)\cdots \vop{p_k}(z_k)=
  \prod_{i<j}(z_i-z_j)^{p_ip_j}\cdot
  \ee^{\sum_{i=1}^kp_i\hat a}
  \prod_{i=1}^k z_i^{p_i a_0}\cdot
  \prod_{m\ge1}\exp\left(\frac{a_{-m}}{m}\sum_{i=1}^k p_i z_i^m\right)
  \exp\left(-\frac{a_{m}}{m}\sum_{i=1}^kp_i z_{i}^{-m}\right).
\end{equation}

We now define the notion of a screening field.
\begin{defn}
A \emph{screening field} $\scr(w)$ for a \voa{} $\VOA{V}$ is a field, generally not belonging to $\VOA{V}$ itself, which has the property that the singular terms of the \opes{} of each of the generating fields of $\VOA{V}$ with $\scr(w)$ are total derivatives in $w$.
\end{defn}
\noindent Our definition of a \emph{screening operator}, in a moral sense at least, is then a \voa{} module homomorphism that can be constructed from screening fields.  This means that a screening operator commutes with the fields, and hence the mode algebra, of the \voa{}.  The standard way of constructing screening operators is as the residues (zero modes) of screening fields; these are guaranteed to commute with all the \voa{} fields, provided that the residues are well defined.  Their chief application stems from the fact that they map (relaxed) \hwvs{} to (relaxed) \hwvs{} and thereby explicitly construct (relaxed) singular vectors.  We remark that if the screening field is a Virasoro primary (excluding the identity field), then its \ope{} with the energy-momentum tensor forces its conformal weight to be $1$.

The \opes{} of the free field $\AKMA{sl}{2}$ fields \eqref{eq:ffrealisation} with a free boson vertex operator are easily computed to be
\begin{equation}
    e(z)\vop{p}(w)\sim 0, \qquad
    h(z)\vop{p}(w)\sim\frac{\alpha p\vop{p}(w)}{z-w},\qquad
    f(z)\vop{p}(w)\sim\frac{\alpha p\vop{p}(w)\gamma(w)}{z-w}.
\end{equation}
For a vertex operator to be a non-trivial screening field, its Heisenberg weight would have to satisfy $\frac{1}{2} p \tbrac{p + \frac{2}{\alpha}} = 1$. However, the above \opes{} show that these vertex operators are not screening fields for $\AKMA{sl}{2}$.  However, the field \(\scr(z)=\vop{-2/\alpha}(z)\beta(z)\) is a screening field \cite{BerFoc90}:
\begin{equation}\label{eq:scrope}
    e(z)\scr(w)\sim 0, \qquad
    h(z)\scr(w)\sim 0, \qquad
    f(z)\scr(w)\sim-t\pd_w\frac{V_{-2/\alpha}(w)}{z-w}.
\end{equation}

It follows that the zero mode
\begin{equation}
\scrs{1} = \oint_0 \scr(z) \: \frac{\dd z}{2 \pi \ii} = \oint_0 \vop{-2/\alpha}(z) \beta(z) \: \frac{\dd z}{2 \pi \ii}
\end{equation}
is a screening operator, whenever the contour around $0$ actually closes.  This will be the case when $\scrs{1}$ acts on a state for which the relevant Fourier expansion of $\scr(z)$ has only integer powers of $z$.  Equivalently, $\scrs{1}$ has a well defined action on a given state if and only if the \ope{} of $\scr(z)$ with the corresponding field is a Laurent series.  For example, $\scrs{1}$ only acts on $\wv{p} = \hv{p} \otimes \gvac$, the tensor product of a Heisenberg vacuum $\hv{p}$ with the ghost vacuum $\gvac$, when $p=\frac{1}{2} m \alpha$, $m \in \ZZ$, because
\begin{equation}
\scr(z) \vop{p}(w) = \vop{-2/\alpha}(z) \vop{p}(w) \beta(z) = \frac{\vop{p-2/\alpha}(w) \beta(w)}{(z-w)^{2p/\alpha}} + \cdots.
\end{equation}
This shows that the screening field $\scr(z)$ only defines module homomorphisms, hence constructs singular vectors, for certain \voa{} modules.

To construct more module homomorphisms, it is natural to consider products of screening fields.  To check that such products also yield screening operators, it is convenient to use the language of differential forms.  Suppose then that $j(z)$ is a \voa{} field and that $\scr_i(w_i)$ is a screening field, so that
\begin{equation}
\comm{j_n}{\scr_i(w_i) \: \dd w_i} = \pd_{w_i} \antiscr_i(w_i) \: \dd w_i = \dd \antiscr_i(w_i),
\end{equation}
for some ($n$-dependent) $\antiscr_i(w_i)$.  Then,
\begin{align}
\comm{j_n}{\scr_1(w_1) \: \dd w_1 \wedge \scr_2(w_2) \: \dd w_2} &= \dd \antiscr_1(w_1) \wedge \scr_2(w_2) \: \dd w_2 + \scr_1(w_1) \: \dd w_1 \wedge \dd \antiscr_2(w_2) \notag \\
&= \dd \tbrac{\antiscr_1(w_1) \: \scr_2(w_2) \: \dd w_2 - \scr_1(w_1) \: \dd w_1 \: \antiscr_2(w_2)}
\end{align}
is exact, hence the commutator vanishes upon integrating over a closed cycle.  The generalisation to more than two screening fields is immediate.

We therefore ask the question of when there exists a closed cycle over which some given product of the screening fields $\scr(z) = \vop{-2/\alpha}(z) \beta(z)$ can be integrated to define \(\AKMA{sl}{2}\)-homomorphisms.  We introduce the shorthand
\begin{equation}
\scrs{r}(z)=\scrs{r}(z_1,\dots,z_k)=\scr(z_1)\cdots\scr(z_r)
\end{equation}
and consider the action of this product of screening fields on the Wakimoto module \(\Wak{p}\) whose Heisenberg \hw{} is $p$ .  Using \eqref{eq:CompVerOps}, this action can be written in the form
\begin{multline}\label{eq:scrprod}
    \Bigl.\scrs{r}(z)\Bigr\rvert_{\Wak{p}}=\Bigl.\prod_{1\le i< j\le r}(z_i-z_j)^{4/\alpha^2}\cdot
    \prod_{i=1}^r z_i^{-2p/\alpha}\\
    \cdot\prod_{m\ge 1}\exp\brac{-\frac{2}{\alpha}\fpowsum{m}{z}\frac{a_{-m}}{m}}
    \exp\brac{\frac{2}{\alpha}\overline{\fpowsum{m}{z}}\frac{a_{m}}{m}}\cdot
    \prod_{i=1}^r\beta(z_i)\Bigr\rvert_{\Wak{p-2r/\alpha}},
\end{multline}
where \(\fpowsum{m}{z}\) denotes a power sum and the overline \(\overline{\fpowsum{m}{z}}=\fpowsum{m}{z_1^{-1},\dots,z_r^{-1}}\) denotes variable inversion (see \appref{sec:SymmPoly} for our conventions for power sums and other symmetric polynomials).  The action of $\scrs{r}(z)$ on the other Wakimoto modules with Heisenberg \hw{} $p$, for example the \rhwm{} $\WakRel{p;q}$, is identical --- the ghost weight $q$ is not changed.

Up to an unimportant phase factor, which we suppress, the first two factors on the \rhs{} of \eqref{eq:scrprod} are
\begin{equation} \label{eq:scrprod'}
  \prod_{1\le i< j\le r}(z_i-z_j)^{4/\alpha^2}\cdot\prod_{i=1}^r z_i^{-2p/\alpha}
  =\prod_{1\le i\neq j\le r}\Bigl(1-\frac{z_i}{z_j}\Bigr)^{1/t}\cdot
  \prod_{i=1}^rz_i^{(r-1)2/\alpha^2-2p/\alpha},
\end{equation}
where we recall that $\alpha^2=2t$.  The second factor on the \rhs{} of this expression therefore isolates all the (potential) non-integer powers of the $z_i$ in \eqref{eq:scrprod}, hence a closed cycle over which $\scrs{r}(z)$ may be integrated will exist precisely when the common exponent of the \(z_i\) in this factor is an integer, \(s \in \ZZ\) say.  This requires the Heisenberg weight \(p\) of the Wakimoto module $\Wak{}$ to have the form
\begin{equation} \label{eq:p_rs}
  p_{r,s}=\frac{r-1}{\alpha}-\frac{s\alpha}{2}, \qquad r\in\mathbb{Z}_{\ge 1},\quad s\in\mathbb{Z}.
\end{equation}
We remark that the multivalued function
\begin{equation} \label{eq:DefG}
  G_r(z;t)=\prod_{1\le i\neq j\le r}\Bigl(1-\frac{z_i}{z_j}\Bigr)^{1/t},
\end{equation}
appearing on the \rhs{} of \eqref{eq:scrprod'}, is just the integration kernel of the inner product \eqref{eq:innerprod} with respect to which the Jack symmetric polynomials are orthogonal.  Setting $p=p_{r,s}$, \eqref{eq:scrprod} now takes the form
\begin{equation} \label{eq:ScrProd}
    \Bigl.\scrs{r}(z)\Bigr\rvert_{\Wak{p_{r,s}}}
    =\Bigl.G_r(z;t)\prod_{i=1}^r z_i^s\cdot
    \prod_{m\ge 1}\exp\brac{-\frac{2}{\alpha}\fpowsum{m}{z}\frac{a_{-m}}{m}}
    \exp\brac{\frac{2}{\alpha}\overline{\fpowsum{m}{z}}\frac{a_{m}}{m}}\cdot
    \prod_{i=1}^r\beta(z_i)\Bigr\rvert_{\Wak{p_{-r,s}}}.
\end{equation}
Again, the action on $\WakRel{p_{r,s};q}$ is identical except that the \rhs{} now acts on $\WakRel{p_{-r,s};q}$.

\subsection{Singular vectors}

The existence of cycles over which the product $\scrs{r}(z)$ of screening fields may be integrated follows from the same arguments used in the analogous question for the free field realisation of the universal Virasoro \voa{}, because the multivalued function \(G_r(z;t)\) is the same in both cases.  This question was answered for the Virasoro case (see \thmref{thm:tkthm}) by Tsuchiya and Kanie \cite{TsuFoc86} to whom we refer for further details.  We will use the cycles $[\Delta_r]$ that they construct in what follows, but normalised so that
\begin{equation}
  \int_{[\Delta_r]}G_r(z;t)\:\frac{\dd z_1\cdots\dd z_r}{z_1\cdots z_r}=1.
\end{equation}

We mention that there are various explicit constructions of cycles, over which screening operators can be integrated, in the \cft{} literature, but that the symmetric function literature uses different constructions again. However, Cohen and Varchenko \cite{CohCyc03} showed that, up to normalisation, there is only one non-trivial homology class of cycles when the integrand is $G_r(z;t)$ times a symmetric function, bar some restrictions on \(t\), and thus the various constructions in the literature are all essentially equivalent.\footnote{It is because of these minor restrictions on \(t\) that we state our choice of class of cycle explicitly. These restrictions are only relevant for the precise statement of \thmref{thm:ConstructSVs}. A more detailed understanding of these cycles is not required for reading the remainder of this article.}

The normalised cycles \([\Delta_r]\) let us construct $\AKMA{sl}{2}$-homomorphisms (screening operators) from the products $\scrs{r}(z)$ of screening fields.  Our aim is to use these homomorphisms to explicitly construct singular vectors in (relaxed) \hwms{} over $\AKMA{sl}{2}$ using symmetric function technology (we refer to \appref{sec:SymmPoly} for a primer on what is needed).  Of course, we need to verify that the screening operators are not just zero maps.

To this purpose, we introduce an algebra isomorphism \(\rho_\delta\), for each \(\delta\in\mathbb{C}^\ast\), from the algebra $\symfunc$ of symmetric polynomials in infinitely many variables to the universal enveloping algebra $\UEA \HA_-$ of the negative subalgebra of the Heisenberg algebra.  This isomorphism is given, on the power sum generators, by
\begin{equation} \label{eq:DefRho}
\rho_{\delta}(\fpowsum{m}{y})=\delta a_{-m}.
\end{equation}
For convenience, we will always take $\rho_{\delta}$ to act only on symmetric polynomials of the variables $y_i$.  For the ghosts, we analogously define injective linear maps $\sigma_r$ and $\widetilde{\sigma}_r$ from the algebra $\sympol{r}$ of symmetric polynomials in $r$ variables to the ghost \uea{} $\UEA \GA$.  These maps are defined on the basis $\set{\dmonsym{\nu}{t} \st \ell(\nu)\le r}$ dual to the symmetric monomials $\monsym{\nu}$ by
\begin{equation} \label{eq:DefSigmas}
\sigma_r(\fdmonsym{\nu}{t}{x})=\beta_{-\nu_1-1}\cdots \beta_{-\nu_r-1},\qquad
\widetilde{\sigma}_r(\fdmonsym{\nu}{t}{x})=\beta_{-\nu_1}\cdots\beta_{-\nu_r}.
\end{equation}
We will only take these maps to act on symmetric polynomials in the $x_i$.  Here, the partition $\nu$ may have length $\ell(\nu)$ strictly less than $r$; in this case, we pad the partition with zeroes so that the images $\sigma_r(\dmonsym{\nu}{t})$ and $\widetilde{\sigma}_r(\dmonsym{\nu}{t})$ are padded by $\beta_{-1}$ or $\beta_0$, respectively, so that the right hand sides of \eqref{eq:DefSigmas} each consist of \(r\) factors.

With these maps, we can prove that the $\scrs{r}(z)$ yield non-trivial screening operators and derive explicit formulae for certain (relaxed) singular vectors of the (relaxed) Verma modules over $\AKMA{sl}{2}$.

\begin{thm} \label{thm:ConstructSVs}
  Let \(r\in\mathbb{Z}_{\ge 1}\), \(s\in\mathbb{Z}\) and \(t\in\mathbb{C}^\ast\) and suppose that \(d(d+1)/t\notin\mathbb{Z}\) and \(d(r-d)/t\notin\mathbb{Z}\), for all integers \(d\) satisfying \(1\le d\le r-1\).
  Then,
  \begin{equation}
    \scrs{r}=\int_{[\Delta_r]}\scrs{r}(z_1,\dots,z_r) \: \dd z_1 \cdots \dd z_r
  \end{equation}
  defines non-trivial \(\AKMA{sl}{2}\)-module homomorphisms between (relaxed) Wakimoto modules:
  \begin{equation}
      \scrs{r}\colon\Wak{p_{r,s}}\rightarrow\Wak{p_{-r,s}},\qquad
      \scrs{r}\colon\WakRel{p_{r,s};q}\rightarrow\WakRel{p_{-r,s};q}.
  \end{equation}
  In particular, if $\wv{p_{r,s}} = \hv{p_{r,s}} \otimes \gvac$ and $\wv{p_{r,s}; q} = \hv{p_{r,s}} \otimes \gv{q}$, where $\hv{p}$ denotes the Heisenberg \hwv{} of weight $p$, $\gvac$ denotes the ghost vacuum and $\gv{q}$ denotes a ghost \rhwv{} of weight $q \notin \ZZ$, then
  \begin{subequations}
    \begin{align}
      \scrs{r}\wv{p_{r,s}}&=
      \begin{cases}
        (\rho_{-\alpha}\circ\sigma_r)
        \left(\fdjack{[(-s-1)^r]}{t}{x,y}\right)\wv{p_{-r,s}} \neq 0 & \text{if \(s\le-1\),} \\
        0 & \text{if \(s\ge 0\),}
      \end{cases}\label{eq:SV}\\
      \scrs{r}\wv{p_{r,s};q}&=
      \begin{cases}
        (\rho_{-\alpha}\circ\widetilde{\sigma}_r)
        \left(\fdjack{[-s^r]}{t}{x,y}\right)\wv{p_{-r,s};q} \neq 0 & \text{if \(s\le0\),} \\
        0 \hphantom{(\rho_{-\alpha}\circ\sigma_r)
        \left(\fdjack{[(-s-1)^r]}{t}{x,y}\right)\wv{p_{-r,s}} \neq {}} & \text{if \(s\ge 1\),}
      \end{cases}\label{eq:RSV}
    \end{align}
  \end{subequations}
  where $\djack{\nu}{t}$ denotes the symmetric polynomials dual to the Jack polynomials (see \appref{sec:SymmPoly}).  The $\scrs{r}$-images of $\wv{p_{r,s}}$ ($\wv{p_{r,s};q}$) in $\Wak{p_{-r,s}}$ ($\WakRel{p_{-r,s};q}$) are therefore (relaxed) singular vectors, for $s \le -1$ ($s \le 0$).
\end{thm}
\begin{proof}
  We first show the non-triviality of $\scrs{r}$. Let \(\dwv{p_{-r,s}}\) and \(\dwv{p_{-r,s};q}\) denote the functionals dual to \(\wv{p_{-r,s}}\) and \(\wv{p_{-r,s};q}\), respectively.  Then, it follows from the ghost commutation relations \eqref{eq:BosGhComm} and adjoint \eqref{eq:GhAdj} that
  \begin{subequations}
  \begin{equation}
    \braket{\gv{q}}{\prod_{i=1}^r \beta(z_i) \gamma_{-s}^r\gv{q}} = (-1)^r r! \braket{\gv{q}}{\gv{q}} \prod_{i=1}^r z_i^{-s-1},\qquad \text{\(s\ge1\), \(q\in\mathbb{C}\).}
  \end{equation}
  Replacing the \rhwv{} $\gv{q}$ by the ghost vacuum $\gvac$ gives the same result, but for $s \ge 0$ (because $\beta_0$ annihilates $\gvac$):
  \begin{equation}
    \braket{\gvac}{\prod_{i=1}^r \beta(z_i) \gamma_{-s}^r\gvac} = (-1)^r r! \braket{\gvac}{\gvac} \prod_{i=1}^r z_i^{-s-1},\qquad \text{\(s\ge0\).}
  \end{equation}
  \end{subequations}
  The non-triviality of \(\scrs{r}\) on $\Wak{p_{r,s}}$, for \(s\ge 0\), now follows from \eqref{eq:ScrProd} and the normalisation of \([\Delta_r]\): 
  \begin{subequations}
  \begin{align}
    \dwv{p_{-r,s}}\scrs{r}\gamma_{-s}^r\wv{p_{r,s}} &=\int_{[\Delta_r]}G_r(z;t)\prod_{i=1}^r z_i^s\cdot
    \dwv{p_{-r,s}}\prod_{i=1}^r \beta(z_i) \gamma_{-s}^r\wv{p_{-r,s}}\:\dd z_1\cdots\dd z_r\notag\\
    &=(-1)^rr!\braket{p_{-r,s}}{p_{-r,s}}\int_{[\Delta_r]}G_r(z;t)\:\frac{\dd z_1\cdots z_r}{z_1\cdots z_r}\notag\\
    &=(-1)^r r!\braket{p_{-r,s}}{p_{-r,s}}\neq0.
  \end{align}
  Here, we note that the $a_{-m}$ and $a_{m}$, with $m \ge 1$, appearing in \eqref{eq:ScrProd} annihilate $\dwv{p_{-r,s}}$ and $\wv{p_{r,s}}$, respectively.  A similar calculation gives the relaxed version for $s\ge 1$:
  \begin{equation}
    \dwv{p_{-r,s};q}\scrs{r}\gamma_{-s}^r\wv{p_{r,s};q} = (-1)^r r!\braket{p_{-r,s};q}{p_{-r,s};q} \neq 0.
  \end{equation}
  \end{subequations}
  We note that the conformal weight of $\wv{p_{-r,s}}$ ($\wv{p_{-r,s}; q}$) is greater than that of $\wv{p_{r,s}}$ ($\wv{p_{r,s}; q}$) by $rs$.  It follows that $\scrs{r}\wv{p_{r,s}}=0$ ($\scrs{r}\wv{p_{r,s};q}=0$) for all $s\ge 1$ because $\scrs{r}$ is an $\AKMA{sl}{2}$-module homomorphism and hence it preserves conformal weights.
  
  We settle the non-triviality of $\scrs{r}$ for the remaining values of \(s\) by explicitly computing the image of the (relaxed) \hwvs{} \(\wv{p_{r,s}}\) and \(\wv{p_{r,s};q}\).  In the former case, we obtain
  \begin{subequations}
    \begin{align}
      \scrs{r}\wv{p_{r,s}}
      &= \int_{[\Delta_r]}G_r(z;t)\prod_{i=1}^r z_i^s\cdot
      \prod_{m\ge1}\exp\left(\frac{-2}{\alpha}\fpowsum{m}{z}\frac{a_{-m}}{m}\right)\cdot
      \prod_{i=1}^r\beta(z_i)\wv{p_{-r,s}}\:\dd z_1\cdots \dd z_r\notag\\
      &=\int_{[\Delta_r]}G_r(z;t)\prod_{i=1}^r z_i^s\cdot
      \prod_{m\ge1}\exp\left(\frac{-2}{\alpha}\frac{\fpowsum{m}{z}a_{-m}}{m}\right)
      \sum_{\nu\st\ell(\nu)\le r}\beta_{-\nu_1-1}\cdots\beta_{-\nu_r-1}
      \fmonsym{\nu}{z}\wv{p_{-r,s}}\:\dd z_1\cdots \dd z_r\notag\\
      &=\int_{[\Delta_r]}G_r(z;t)\prod_{i=1}^r z_i^s\cdot
      \rho_{-\alpha}\biggl(\prod_{m\ge1}\exp\left(\frac{1}{t}\frac{\fpowsum{m}{y}\fpowsum{m}{z}}{m}\right)\biggr)\cdot
      \sigma_r\biggl(\sum_{\nu\st\ell(\nu)\le r}\fdmonsym{\nu}{t}{x}\fmonsym{\nu}{z}\biggr)\wv{p_{-r,s}}\:\dd
      z_1\cdots \dd z_r\notag\\
      &=\int_{[\Delta_r]}G_r(z;t)\prod_{i=1}^r z_i^{s+1}\cdot
      (\rho_{-\alpha}\circ\sigma_r)\left(\prod_{m\ge1}\exp\left(\frac{1}{t}\frac{\left(\fpowsum{m}{x}+\fpowsum{m}{y}\right)\fpowsum{m}{z}}{m}\right)\right)
      \wv{p_{-r,s}}\:\frac{\dd z_1\cdots \dd z_r}{z_1\cdots z_r} \notag
    \intertext{(if $s=0$, this vanishes, in agreement with \eqref{eq:SV}, as the powers of the $z_i$ in the integrand are all positive)}
      &=\sum_{\nu\st\ell(\nu)\le r} \intprod{\fjack{[(-1-s)^r]}{t}{z}}{\fjack{\nu}{t}{z}}_t^r
      \:(\rho_{-\alpha}\circ\sigma_r)\left(\fdjack{\nu}{t}{x,y}\right) \wv{p_{-r,s}}\notag\\
      &=(\rho_{-\alpha}\circ\sigma_r)\left(\fdjack{[(-1-s)^r]}{t}{x,y}\right)\wv{p_{-r,s}}.
    \end{align}
    Here, we have used the finite-variable version of \eqref{eq:Cauchy} to rewrite the sum over the symmetric monomials $\monsym{\nu}$ and their duals $\dmonsym{\nu}{t}$ as a product of exponentials of power sums.  Then, we note that $\fpowsum{m}{x}+\fpowsum{m}{y}$ is the power sum $\fpowsum{m}{x,y}$ in both \(x_i\) and \(y_j\) variables and use \eqref{eq:Cauchy} again to expand the product of exponentials in terms of Jack polynomials and their duals.  Finally, we use \propref{prop:jackprops} to identify $\prod_i z_i^{-1-s}$ as a Jack polynomial in the $z_i^{-1}$, then apply the orthogonality of Jack polynomials, and lastly note that the Jack polynomial norm follows from the normalisation of \([\Delta_r]\).  For $s\le -1$, the result is non-vanishing since $\rho_{-\alpha}$ and $\sigma_r$ are injective.  This proves the non-triviality of $\scrs{r}$ on $\Wak{p_{r,s}}$ for all $s \in \ZZ$, as well as \eqref{eq:SV}.

    A similar computation in the relaxed sector results in
    \begin{align}
      \scrs{r}\wv{p_{r,s};q}
      &=\int_{[\Delta_r]}G_r(z;t)\prod_{i=1}^r z_i^s\cdot
      \prod_{m\ge1}\exp\left(\frac{-2}{\alpha}\frac{\fpowsum{m}{z}a_{-m}}{m}\right)\cdot
      \prod_{i=1}^r\beta(z_i)\wv{p_{-r,s};q}\:\dd z_1\cdots \dd z_r\notag\\
      &=\int_{[\Delta_r]}G_r(z;t)\prod_{i=1}^r z_i^s\cdot
      \prod_{m\ge1}\exp\left(\frac{-2}{\alpha}\frac{\fpowsum{m}{z}a_{-m}}{m}\right)\cdot
      \sum_{\nu\st\ell(\nu)\le r}\beta_{-\nu_1}\cdots\beta_{-\nu_r}\fmonsym{\nu}{z}\wv{p_{-r,s};q}\:      \frac{\dd z_1\cdots \dd z_r}{z_1\cdots z_r}\notag\\
      &=(\rho_{-\alpha}\circ\widetilde{\sigma_r})\left(\fdjack{[-s^r]}{t}{x,y}\right)
      \wv{p_{-r,s};q},
    \end{align}
  \end{subequations}
  which is likewise non-vanishing for $s\le 0$ as $\rho_{-\alpha}$ and $\widetilde{\sigma_r}$ are injective.  This proves \eqref{eq:RSV} and the non-triviality of $\scrs{r}$ on $\WakRel{p_{r,s};q}$, for all $s \in \ZZ$.
\end{proof}

Suppose now that the $\AKMA{sl}{2}$ level $k$ is admissible:  $t=k+2=\frac{u}{v}$, where $u \in \ZZ_{\ge 2}$ and $v \in \ZZ_{\ge 1}$ are coprime.  The $\AKMA{sl}{2}$ vacuum may be identified with the Wakimoto \hwv{} $\wv{0} = \wv{p_{1,0}} = \wv{p_{-u+1,-v}} \in \Wak{p_{-u+1,-v}}$ (we note the symmetry $p_{r,s} = p_{r+u,s+v}$).  \thmref{thm:ConstructSVs} then guarantees that the $\AKMA{sl}{2}$-module homomorphism $\scrs{u-1}$ acts non-trivially on $\wv{p_{u-1,-v}}$ to give a non-trivial singular vector in the vacuum module.  The corresponding field then generates the non-trivial proper ideal of the universal \voa{}, proving the following result:
\begin{cor} \label{cor:WakVacIsUniversal}
The universal \voa{} \(\UVOA{k}\) of $\AKMA{sl}{2}$ at non-critical level $k\neq-2$ may be realised as a subalgebra of the Wakimoto \voa{}.
\end{cor}
\noindent This is, of course, obvious if $k$ is not admissible.  What it means
in the admissible level case is that calculations requiring the singular
vector of the vacuum module of $\AKMA{sl}{2}$ may be equivalently carried out
in the free field realisation using Jack symmetric polynomials.

\section{Classifying admissible modules} \label{sec:AdmMod}

In this section, we specialise to admissible levels $k=t-2$, where $t=\frac{u}{v}$ and $u\in\ZZ_{\ge2}$ and $v\in\ZZ_{\ge1}$ are coprime.  Then, the universal \voa{} $\UVOA{k}$ of $\AKMA{sl}{2}$ is not simple and the unique maximal proper ideal is generated by the field that corresponds to the \sv{} of the vacuum module (see \secref{sec:RelEx}).  As we saw in \thmref{thm:ConstructSVs} and \corref{cor:WakVacIsUniversal}, the vacuum module may be constructed as a submodule of the Wakimoto vacuum module $\Wak{0}$ and its singular vector is then explicitly given by $\scrs{u-1} \wv{p_{u-1,-v}}$.

Quotienting by this maximal proper ideal, that is, setting the \sv{} $\scrs{u-1} \wv{p_{u-1,-v}}$ to zero, amounts to replacing the universal \voa{} $\UVOA{k}$ by its simple counterpart \(\AdmMod{u}{v}\).  Our aim in this section is to use the explicit expression for the vacuum \sv{} to classify the possible (relaxed) highest weights and thereby determine the spectrum of $\AdmMod{u}{v}$-modules.  By \corref{cor:WakVacIsUniversal}, these calculations may be performed in Wakimoto's free field realisation.  More specifically, we will use symmetric polynomial technology to compute a generator of the annihilating ideal in Zhu's algebra. For readers unfamiliar with Zhu's algebra, we refer to Appendix \ref{sec:zhu} for motivation, basic definitions and a very short primer.

An old result of Frenkel and Zhu \cite{FreVer92} states that Zhu's algebra \(\zhu{\UVOA{k}}\) for \(\UVOA{k}\) is nothing but the \uea{} $\UEA\:\SLA{sl}{2}$ of (non-affine) \(\SLA{sl}{2}\) (see \propref{prop:zhusl2}).  By \propref{prop:zhuideal}, Zhu's algebra \(\zhu{\AdmMod{u}{v}}\) for the quotient \(\AdmMod{u}{v}\) is then the quotient of \(\UEA\:\SLA{sl}{2}\) by the annihilating ideal generated by the representative of the singular vector \(\scrs{u-1}\wv{p_{u-1,-v}}\). Since Zhu's algebra is filtered by conformal weight, whereas the conformal weight of \(e(z)\), \(h(z)\) and \(f(z)\) is 1 and the conformal weight of the singular vector is \((u-1)v\), it follows that the image of the singular vector in \(\zhu{\UVOA{k}}\) is a polynomial in the \(\SLA{sl}{2}\) generators of total degree at most \((u-1)v\). Furthermore, as Zhu's algebra is just the algebra of zero modes acting on \rhwvs{}, the polynomial corresponding to the singular vector can be determined by evaluating the zero mode of the singular vector on general \rhwvs{}, as in \cite{FeiAnn92}, since this is equivalent to evaluating the polynomial at infinitely many points.

However, the $\SLA{sl}{2}$-weight of the vacuum \sv{} $\scrs{u-1}\wv{p_{u-1,-v}}$ is $\lambda_{u-1,-v} = \lambda_{p_{u-1,-v}} = 2(u-1)$, which means that the corresponding field and its zero mode shifts the \(\SLA{sl}{2}\)-weight of any \rhwv{} upon which it acts by this amount.  It is far more convenient to work with a field that does not shift \(\SLA{sl}{2}\)-weights and so we instead consider the field corresponding to the $\SLA{sl}{2}$-weight $0$ vector
\begin{align}\label{eq:idealgenerator}
  f_0^{u-1}\scrs{u-1}\wv{p_{u-1,-v}} &= \scrs{u-1}f_0^{u-1}\wv{p_{u-1,-v}} \notag \\
  &= \lambda_{u-1,-v} \brac{\lambda_{u-1,-v} - 1} \cdots \brac{\lambda_{u-1,-v} - u+2} \scrs{u-1}\gamma_0^{u-1}\wv{p_{u-1,-v}}.
\end{align}
Since $u \ge 2$, we may renormalise this vector by dividing by the non-zero $\lambda_{u-1,-v}$-dependent factors on the \rhs{}.  The field corresponding to this renormalised vector is then
\begin{align} \label{eq:DefChi}
  \chi(w)&=\int_{[\Delta_r]}\scrs{u-1}(z_1+w,\dots,z_{u-1}+w)\vop{p_{u-1,-v}}(w)\gamma(w)^{u-1} \:\dd z_1\cdots\dd z_{u-1} \notag \\
  &= \int_{[\Delta_r]} \vop{-2/\alpha}(z_1+w)\cdots\vop{-2/\alpha}(z_{u-1}+w)\vop{p_{u-1,-v}}(w) \beta(z_1+w)\cdots\beta(z_{u-1}+w)\gamma(w)^{u-1} \:\dd z_1\cdots\dd z_{u-1}.
\end{align}
We note that $e_0^{u-1}$ acting on the vector $f_0^{u-1}\scrs{u-1}\wv{p_{u-1,-v}}$ gives a non-zero multiple of the \sv{} $\scrs{u-1}\wv{p_{u-1,-v}}$.  It follows that the annihilating ideal that we obtain from $\chi(w)$ will be the same as that which we would have obtained if we had instead worked with the \sv{} directly.

\subsection{Admissible \hwms{}} \label{sec:HWAdmMod}

In this section, we will evaluate the action of the zero mode of \(\chi(w)\) on a general \hwv{} before moving on to the evaluation of \(\chi(w)\) on \rhwvs{}. This \hw{} computation will not yet yield sufficient information to determine the image of \(\chi(w)\) in Zhu's algebra \(\zhu{\UVOA{k}}\) as a polynomial in the \(\SLA{sl}{2}\) generators \(e,h\) and \(f\), but it will constrain the weights of admissible \hwvs{} to a finite set.
 
Every \hwv{} of $\UVOA{k}$ may be realised as a \hwv{} of some Wakimoto module $\Wak{p}$.  It follows that the Heisenberg weights $p$ of the admissible \hwvs{} of $\AKMA{sl}{2}$, that is, the \hwvs{} of the $\AdmMod{u}{v}$-modules, are zeroes of
\begin{multline}
  \dwv{p}\chi(w)\wv{p} = \int_{[\Delta_r]} \inner{\hv{p}}{\vop{-2/\alpha}(z_1+w) \cdots \vop{-2/\alpha}(z_{u-1}+w) \vop{p_{u-1,-v}}(w) \: \hv{p}} \\
  \cdot \inner{\gvac}{\beta(z_1+w) \cdots \beta(z_{u-1}+w) \gamma(w)^{u-1} \: \gvac} \: \dd z_1\cdots\dd z_{u-1}.
\end{multline}
The ghost contribution is easily evaluated by computing the \ope{} of the ghost fields.  Since $\gvac$ is the ghost vacuum, only the fully contracted part of the \ope{} contributes:
\begin{align}
\inner{\gvac}{\beta(z_1+w) \cdots \beta(z_{u-1}+w) \gamma(w)^{u-1} \: \gvac} &= \inner{\gvac}{\normord{\beta(z_1+w) \cdots \beta(z_{u-1}+w)} \normord{\gamma(w) \cdots \gamma(w)} \: \gvac} \notag \\
&= \frac{(-1)^{u-1} (u-1)!}{z_1 \cdots z_{u-1}} \inner{\gvac}{\gvac}.
\end{align}
The contribution from the free boson part of the Wakimoto realisation is likewise easily determined.  Up to non-zero constant factors, which obviously do not affect the zeroes of $\dwv{p}\chi(w)\wv{p}$, this contribution is
\begin{align} \label{eq:FBContribution}
&\inner{\hv{p}}{\vop{-2/\alpha}(z_1+w) \cdots \vop{-2/\alpha}(z_{u-1}+w) \vop{2(u-1)/\alpha}(w) \: \hv{p}} \notag \\&\mspace{200mu} = \prod_{1 \le i \neq j \le u-1} (z_i-z_j)^{1/t} \cdot \prod_{i=1}^{u-1} (z_i+w)^{-2p/\alpha} z_i^{-2(u-1)/t} \cdot w^{2(u-1)p/\alpha} \inner{\hv{p}}{\hv{p}} \notag \\
&\mspace{200mu} = G_{u-1}(z;t) \prod_{i=1}^{u-1} z_i^{-v} \cdot \prod_{i=1}^{u-1} \brac{1 + \frac{z_i}{w}}^{-2p/\alpha} \inner{\hv{p}}{\hv{p}},
\end{align}
where we have used \eqref{eq:CompVerOps}, noting that all the terms involving the $a_n$ with $n \neq 0$ either annihilate the Heisenberg vacuum $\hv{p}$ or its dual, and recognised the kernel of the symmetric polynomial inner product from \eqref{eq:DefG}.

Putting these contributions together, we find that the admissible Heisenberg weights of an $\AKMA{sl}{2}$ \hwv{} (in the Wakimoto free field realisation) must satisfy
\begin{equation} \label{eq:HWConstraint}
0 = \dwv{p}\chi(w)\wv{p} = \int_{[\Delta_r]} G_{u-1}(z;t) \: \overline{\fjack{[v^{u-1}]}{t}{z}} \: \prod_{i=1}^{u-1} \brac{1 + \frac{z_i}{w}}^{-2p/\alpha} \: \frac{\dd z_1\cdots\dd z_{u-1}}{z_1 \cdots z_{u-1}} \braket{p}{p},
\end{equation}
where we have recognised the product of the $z_i$ as a Jack polynomial using \propref{prop:jackprops} (recall that the overline indicates a symmetric polynomial in the inverse variables $z_i^{-1}$).  We have also, again, neglected an overall non-zero constant factor.  This expression has the form of an inner product for symmetric polynomials; to evaluate it, we only need to decompose the product over $i$ into Jack polynomials.  For this, we use specialisation in the form given in \eqref{eq:SpecialisingToJacks} with $x_i = -z_i/w$ for $i \le u-1$, $x_i=0$ for $i>u-1$, and $X=\alpha p$:
\begin{equation} \label{eq:FirstSpecialisation}
\prod_{i=1}^{u-1} \brac{1 + \frac{z_i}{w}}^{-2p/\alpha} = \sum_{\tau} \fjack{\tau}{t}{-z_1/w, \ldots, -z_{u-1}/w, 0, 0, \ldots} \: \func{\Xi_{\alpha p}}{\fdjack{\tau}{t}{y}} = \sum_{\tau} \frac{(-1)^{\abs{\tau}}}{w^{\abs{\tau}}} \fjack{\tau}{t}{z} \: \func{\Xi_{\alpha p}}{\fdjack{\tau}{t}{y}}.
\end{equation}
Here, we have also used the homogeneity of the Jack polynomials ($\abs{\tau}$ is the sum of the parts of the partition $\tau$). 

Using the orthogonality of Jack polynomials, \eqref{eq:HWConstraint} now becomes
\begin{align}
0 &= \sum_{\tau} w^{-\abs{\tau}} \intprod{\jack{[v^{u-1}]}{t}}{\jack{\tau}{t}}_t^{u-1} \func{\Xi_{\alpha p}}{\fdjack{\tau}{t}{y}} = w^{-\abs{[v^{u-1}]}} \intprod{\jack{[v^{u-1}]}{t}}{\jack{[v^{u-1}]}{t}}_t^{u-1} \func{\Xi_{\alpha p}}{\fdjack{[v^{u-1}]}{t}{y}} \notag \\
&= w^{-(u-1)v} \: \func{\Xi_{\alpha p}}{\fdjack{[v^{u-1}]}{t}{y}} = w^{-(u-1)v} \prod_{b \in [v^{u-1}]} \frac{\alpha p + t a'(b) - l'(b)}{t \tbrac{a(b)+1} + l(b)} \notag \\
&= w^{-(u-1)v} \prod_{i=1}^{u-1} \prod_{j=1}^{v} \frac{\alpha p + t (j-1) - (i-1)}{t \tbrac{v-j+1} + u-1-i},
\end{align}
where the result of specialising the dual Jack polynomials $\djack{\tau}{t}$ is given in \propref{prop:jackprops} (or \eqref{eq:SpecialisingToJacks}) and the arm and leg (co)lengths of \eqref{eq:ArmLeg} are easy to determine for the rectangular partition $[v^{u-1}]$.  As the denominators of the factors appearing in this expression are all strictly positive, we arrive at the constraint
\begin{equation}
\prod_{i=1}^{u-1} \prod_{j=0}^{v-1} \tbrac{\alpha p + tj - (i-1)} = 
\prod_{i=1}^{u-1} \prod_{j=0}^{v-1} \tbrac{\alpha p - \alpha p_{i,j}} = 0.
\end{equation}
This proves the following result:
\begin{prop} \label{prop:ClassHWMods}
Let $\lambda_{r,s} = \lambda_{p_{r,s}} = \alpha p_{r,s} = r-1-ts$.  Then, every \hwv{} of an $\AdmMod{u}{v}$-module has $\SLA{sl}{2}$-weight of the form $\lambda_{r,s}$, where $r=1, 2, \ldots, u-1$ and $s=0, 1, \ldots, v-1$.
\end{prop}
\noindent Note that when $v=1$, so that $k=t-2=u-2 \in \ZZ_{\ge 0}$, the allowed $\SLA{sl}{2}$-weights $\lambda_{r,s} = \lambda_{r,0} = r-1$ belong to the set $\set{0, 1, \ldots, k}$.  These are, of course, the \hw{}s of the integrable $\AKMA{sl}{2}$-modules and are well known to be the $\AdmMod{k+2}{1}$-modules that arise in the \WZW{} models defined on the Lie group $\SLG{SU}{2}$.

At this point, we cannot say whether all these \hwvs{} do actually appear in an $\AdmMod{u}{v}$-module.  For this, we need to work out the constraints on an arbitrary \rhwv{} because it is these constraints which allow us to write down the generator of the annihilating ideal of Zhu's algebra.  It is, however, well known that the Verma module $\Ver{\lambda_{r,s}}$, with $r=1, 2, \ldots, u-1$ and $s=0, 1, \ldots, v-1$, has infinitely many linearly independent \svs{}; however, none have $\SLA{sl}{2}$-weights belonging to the allowed set except the generator of $\SLA{sl}{2}$-weight $\lambda_{r,s}$.  It follows that there are only finitely many \hw{} $\AdmMod{u}{v}$-modules:
\begin{cor} \label{cor:ClassHWMods}
Every \hw{} $\AdmMod{u}{v}$-module is isomorphic to one of the simple $\AKMA{sl}{2}$-modules $\Irr{\lambda_{r,0}}$ or $\Disc{\lambda_{r,s}}$, where $r=1, 2, \ldots, u-1$ and $s=1, 2, \ldots, v-1$.
\end{cor}
\noindent Again, we have not yet proven that all these $\UVOA{k}$-modules are actually $\AdmMod{u}{v}$-modules.  However, it is germane to point out, at this point, that these \hwms{} are precisely the admissible modules first discussed by Kac and Wakimoto \cite{KacMod88}.

\subsection{Admissible \rhwms{}} \label{sec:HWRelAdmMod}

We now turn to the more intricate, but ultimately more rewarding, analysis of the \rhwms{} of $\AdmMod{u}{v}$.  The goal is to determine the image of the field \(\chi(w)\) in Zhu's algebra \(\zhu{\UVOA{k}} \cong \UEA\:\SLA{sl}{2}\) as a polynomial \(I_{u,v}(e,h,f)\) in the \(\SLA{sl}{2}\) generators. Here, we identify the $\SLA{sl}{2}$ generators $e$, $h$ and $f$ with the images of $e(z)$, $h(z)$ and $f(z)$, respectively, in $\zhu{\UVOA{k}}$. Since Zhu's algebra is nothing but the algebra of zero modes acting on \rhwvs{} (\appref{sec:zhu}), $I_{u,v}(e,h,f)$ is identical to the polynomial \(I_{u,v}(e_0,h_0,f_0)\) that describes the action of the zero mode \(\chi_0\) on \rhwvs{}.  We remark that as the $\SLA{sl}{2}$-weight of $\chi(w)$ is $0$, the polynomial $I_{u,v}$ may be expressed as a polynomial in $h_0$ and the Virasoro zero mode $L_0$.

So as in the previous section, we evaluate matrix elements containing the field \(\chi(w)\), but this time the ``bra'' and the ``ket'' will be \rhwvs{} from a general relaxed Wakimoto module:
\begin{multline} \label{eq:RelaxedConstraint}
  \dwv{p;q}\chi(w)\wv{p;q} = \int_{[\Delta_r]} \inner{\hv{p}}{\vop{-2/\alpha}(z_1+w) \cdots \vop{-2/\alpha}(z_{u-1}+w) \vop{p_{u-1,-v}}(w) \: \hv{p}} \\
  \cdot \inner{\gv{q}}{\normord{\beta(z_1+w) \cdots \beta(z_{u-1}+w)} \normord{\gamma(w)^{u-1}} \: \gv{q}} \: \dd z_1\cdots\dd z_{u-1}.
\end{multline}
We recall that $\gv{q}$ satisfies $J_0 \gv{q} = \gamma_0 \beta_0 \gv{q} = q \gv{q}$.

The contribution from the free boson is exactly the same as in the non-relaxed case and was given in \eqref{eq:FBContribution}.  The ghost contribution, however, requires more work.  Wick's theorem lets us write this contribution in terms of contractions and normally-ordered products:
\begin{multline}
\inner{\gv{q}}{\normord{\beta(z_1+w) \cdots \beta(z_{u-1}+w)} \normord{\gamma(w)^{u-1}} \: \gv{q}} \\
= \sum_{I\subseteq\{1,\dots,u-1\}} \frac{(-1)^{\abs{I}} (u-1)!}{(u-1-\abs{I})!} \prod_{i \in I} z_i^{-1} \cdot \inner{\gv{q}}{\normord{\displaystyle \prod_{i \notin I} \beta(z_i+w) \cdot \gamma(w)^{u-1-\abs{I}}} \: \gv{q}}
\end{multline}
Here, each factor of $-z_i^{-1}$ is the contraction of $\beta(z_i+w)$ and $\gamma(w)$ and the factorials count how many such contractions are needed.  As \rhwvs{} are not necessarily annihilated by $\beta_0$, the normally-ordered factor is quite non-trivial:
\begin{align}
\inner{\gv{q}}{\normord{\displaystyle \prod_{i \notin I} \beta(z_i+w) \cdot \gamma(w)^{u-1-\abs{I}}} \: \gv{q}} &= \inner{\gv{q}}{\gamma_0^{u-1-\abs{I}} \beta_0^{u-1-\abs{I}} \: \gv{q}} \prod_{i \notin I} (z_i+w)^{-1} \notag \\
&= (u-1-\abs{I})! \binom{u-2-\abs{I}+q}{u-1-\abs{I}} \prod_{i \notin I} (z_i+w)^{-1} \inner{\gv{q}}{\gv{q}}.
\end{align}
Up to an overall non-zero constant factor, the total ghost contribution is therefore
\begin{align}
&\sum_{I\subseteq\{1,\dots,u-1\}} (-1)^{\abs{I}} \binom{u-2-\abs{I}+q}{u-1-\abs{I}} \prod_{i \in I} z_i^{-1} \cdot \prod_{i \notin I} (z_i+w)^{-1} \notag \\
&\mspace{100mu} = \prod_{i=1}^{u-1} (z_i+w)^{-1} \cdot \sum_{n=0}^{u-1} (-1)^n \binom{u-2-n+q}{u-1-n}
\felsym{n}{1 + \frac{w}{z_1},\dots,1 + \frac{w}{z_{u-1}}} \notag \\
&\mspace{100mu} = \prod_{i=1}^{u-1} \Bigl( 1 + \frac{z_i}{w} \Bigr)^{-1} \cdot \sum_{n=0}^{u-1} (-1)^n \binom{u-2-n+q}{u-1-n} \sum_{m=0}^n \binom{u-1-m}{n-m} \: \overline{\felsym{m}{z}} \: w^{m-(u-1)},
\end{align}
where $\elsym{m}$ denotes the $m$-th elementary symmetric polynomial and we have used the identity
\begin{equation}
  \felsym{n}{1+x_1,\dots,1+x_{u-1}}=\sum_{m=0}^n\binom{u-1-m}{n-m}\felsym{m}{x}
\end{equation}
to get from the second to the third line.

Combining this with the free boson contribution \eqref{eq:FBContribution}, the matrix element \eqref{eq:RelaxedConstraint} is thus proportional to
\begin{align}
&\int_{[\Delta_r]} G_{u-1}(z;t) \prod_{i=1}^{u-1} z_i^{-(v-1)} \cdot \prod_{i=1}^{u-1} \brac{1 + \frac{z_i}{w}}^{-2p/\alpha - 1} \notag \\
&\mspace{100mu} \cdot \sum_{n=0}^{u-1} (-1)^n \binom{u-2-n+q}{u-1-n} \sum_{m=0}^n \binom{u-1-m}{n-m} \: \overline{\felsym{m}{z}} \: w^m \: \frac{\dd z_1\cdots\dd z_{u-1}}{z_1 \cdots z_{u-1}} \notag \\
&= \sum_{n=0}^{u-1} (-1)^n \binom{u-2-n+q}{u-1-n} \sum_{m=0}^n \binom{u-1-m}{n-m} w^m \intprod{\fjack{[(v-1)^{u-1}]}{t}{z} \fjack{[1^m]}{t}{z}}{\prod_{i=1}^{u-1} \brac{1 + \frac{z_i}{w}}^{-2p/\alpha - 1}}_t^{u-1},
\end{align}
where we recall that elementary symmetric polynomials are examples of Jack polynomials (\propref{prop:jackprops}).  Using \eqref{eq:JackByRectJack} and specialisation as in \eqref{eq:FirstSpecialisation}, but with $X = \alpha p + t$, this reduces to
\begin{align} \label{eq:RelaxedConstraint'}
&\sum_{n=0}^{u-1} (-1)^n \binom{u-2-n+q}{u-1-n} \sum_{m=0}^n \binom{u-1-m}{n-m} w^m \sum_{\tau} \frac{(-1)^{\abs{\tau}}}{w^{\abs{\tau}}} \intprod{\fjack{[(v-1)^{u-1}] + [1^m]}{t}{z}}{\fjack{\tau}{t}{z}}_t^{u-1} \func{\Xi_{\alpha p + t}}{\fdjack{\tau}{t}{y}} \notag \\
&= \sum_{n=0}^{u-1} (-1)^n \binom{u-2-n+q}{u-1-n} \sum_{m=0}^n \binom{u-1-m}{n-m} \frac{(-1)^{\abs{\mu}}}{w^{\abs{\mu} - m}} \intprod{\fjack{\mu}{t}{z}}{\fjack{\mu}{t}{z}}_t^{u-1} \func{\Xi_{\alpha p + t}}{\fdjack{\mu}{t}{y}} \notag \\
&= w^{-(u-1)(v-1)} \sum_{n=0}^{u-1} \binom{u-2-n+q}{u-1-n} \sum_{m=0}^n (-1)^{m+n} \binom{u-1-m}{n-m} \intprod{\fjack{[1^m]}{t}{z}}{\fjack{[1^m]}{t}{z}}_t^{u-1} \func{\Xi_{\alpha p + t}}{\fdjack{\mu}{t}{y}},
\end{align}
where $\mu = [(v-1)^{u-1}] + [1^m] = [v^m, (v-1)^{u-1-m}]$ and $\abs{\mu} = (u-1)(v-1)+m$.  In the last step, we have used \eqref{eq:JackByRectJack} again and the definition \eqref{eq:innerprod} of the symmetric polynomial inner product.

\propref{prop:jackprops} gives the norm squared of $\jack{[1^m]}{t} = \elsym{m}$ and the specialisation of $\djack{\mu}{t}$ as
\begin{subequations}
\begin{align}
\intprod{\fjack{[1^m]}{t}{z}}{\fjack{[1^m]}{t}{z}}_t^{u-1} &= \prod_{i=1}^m \frac{(u-i)(t+m-i)}{(m-i+1)(t+u-1-i)} = \binom{u-1}{m} \prod_{i=1}^m \frac{t+m-i}{t+u-1-i}, \\
\func{\Xi_{\alpha p + t}}{\fdjack{\mu}{t}{y}} &= \prod_{j=1}^{v-1} \sqbrac{\prod_{i=1}^m \frac{\alpha p + tj-i+1}{2u-1-i-t(j-1)} \cdot \prod_{i=m+1}^{u-1} \frac{\alpha p + tj-i+1}{2u-1-i-tj}} \cdot \prod_{i=1}^m \frac{\alpha p + u-i+1}{t+m-i} \notag \\
&= \frac{\displaystyle \prod_{r=1}^{u-1} \prod_{s=1}^{v-1} (\lambda_p - \lambda_{r,s}) \cdot \prod_{i=0}^{m-1} (\alpha p +u-i)}{\displaystyle \prod_{i=1}^{u-1} \prod_{j=2}^{v-1} \tbrac{2u-1-i-t(j-1)} \cdot \prod_{i=0}^{m-1} \tbrac{2(u-1)-i} \cdot \prod_{i=m+1}^{u-1} (t+u-1-i) \cdot \prod_{i=1}^m (t+m-i)}. \label{eq:RelaxedSpecialisation}
\end{align}
\end{subequations}
Here, we have assumed that $v>1$; if $v=1$, then the denominator of \eqref{eq:RelaxedSpecialisation} is just $\prod_{i=1}^m (t+m-i)$.  Noting that the double product in the denominator of \eqref{eq:RelaxedSpecialisation} is a constant, independent of $m$, $n$, $p$ and $q$, the matrix element \eqref{eq:RelaxedConstraint} is thus proportional to
\begin{multline}
\prod_{r=1}^{u-1} \prod_{s=1}^{v-1} (\lambda_p - \lambda_{r,s}) \cdot \sum_{n=0}^{u-1} \binom{u-2-n+q}{u-1-n} \sum_{m=0}^n (-1)^{m+n} \binom{u-1-m}{n-m} \binom{u-1}{m} \prod_{i=0}^{m-1} \frac{\alpha p + u-i}{2(u-1)-i} \\
= \prod_{r=1}^{u-1} \prod_{s=1}^{v-1} (\lambda_p - \lambda_{r,s}) \cdot \sum_{\ell = 0}^{u-1} \binom{q-1}{\ell} \binom{u-1+\ell}{u-1} \binom{\alpha p + u}{u-1-\ell},
\end{multline}
where in addition to suppressing the denominator we have also suppressed an overall non-zero constant factor that arises when simplifying the binomial expressions.

The double product in this expression can be interpreted by noting that
\begin{equation}
(\lambda_p - \lambda_{r,s}) (\lambda_p - \lambda_{u-r,v-s}) = 4t(\Delta_p - \Delta_{r,s}), \qquad \Delta_{r,s} = \Delta_{p_{r,s}} = \frac{(r-ts)^2-1}{4t} = \frac{(vr-us)^2-v^2}{4uv},
\end{equation}
by \eqref{eq:RelaxedWts}.  If we define $K(u,v)$ to be the set of pairs $(r,s) \in \set{1, \ldots, u-1} \times \set{1, \ldots, v-1}$ with $(r,s)$ and $(u-r,v-s)$ identified, 
then the matrix element \eqref{eq:RelaxedConstraint} is given by
\begin{equation} \label{eq:RelaxedConstraint''}
\dwv{p;q}\chi(w)\wv{p;q}=\text{const}\cdot
\prod_{(r,s) \in K(u,v)} (\Delta_p - \Delta_{r,s}) \cdot \sum_{\ell = 0}^{u-1} \binom{q-1}{\ell} \binom{u-1+\ell}{u-1} \binom{\alpha p + u}{u-1-\ell}.
\end{equation}

In order to use this to find a generator of the annihilating ideal in Zhu's algebra, the sum factor in \eqref{eq:RelaxedConstraint''} needs to be expressible in terms of \(\SLA{sl}{2}\) data. To demonstrate this, we define a function $f$ of $u$, $p$ and $q$ by
\begin{equation} \label{eq:DefF}
f_{p;q}(u) = \sum_{\ell = 0}^{u-1} \binom{q-1}{\ell} \binom{u-1+\ell}{u-1} \binom{\alpha p + u}{u-1-\ell}.
\end{equation}
For small values of $u$, it gives polynomials in $\lambda_{p;q}$ and $\Delta_p$:
\begin{equation} \label{eq:SumFactorExs}
f_{p;q}(2) = \lambda_{p;q}, \qquad 
f_{p;q}(3) = \frac{3}{4} \lambda_{p;q}^2 - t \Delta_p, \qquad 
f_{p;q}(4) = \frac{5}{12} \lambda_{p;q}^3 + \brac{\frac{1}{3} - t \Delta_p} \lambda_{p;q}.
\end{equation}
Of course, $\alpha = \sqrt{2t} = \sqrt{2u/v}$ also depends upon $u$, but may be regarded as an independent variable for the following analysis because of its $v$-dependence.
\begin{prop}
For each $u \in \ZZ_{\ge 0}$, $f_{p;q}(u)$ is a polynomial in $\lambda_{p;q}$ and $\Delta_p$ that satisfies the recursion relation
\begin{equation} \label{eq:Recurse}
f_{p;q}(u+2) = \frac{(2u+1) \lambda_{p;q}}{(u+1)^2} f_{p;q}(u+1) - \frac{4t \Delta_p - (u-1)(u+1)}{(u+1)^2} f_{p;q}(u).
\end{equation}
\end{prop}
\begin{proof}
We apply Zeilberger's creative telescoping algorithm, see \cite{ZeiA=B96} for background.  If we set
\begin{equation}
F_{p;q}(u,\ell) = \binom{q-1}{\ell} \binom{u-1+\ell}{u-1} \binom{\alpha p + u}{u-1-\ell},
\end{equation}
then the algorithm constructs the recursion relation
\begin{equation}
\tbrac{(\alpha p + 1)^2 - u^2} F_{p;q}(u,\ell) - (2u+1) \lambda_{p;q} F_{p;q}(u+1,\ell) + (u+1)^2 F_{p;q}(u+2,\ell) = G(u,\ell+1) - G(u,\ell),
\end{equation}
where $G(u,\ell) = R(u,\ell) F_{p;q}(u,\ell)$ and
\begin{equation}
R(u,\ell) = -\frac{2 \ell^2 (\alpha p + 1 + \ell) \tbrac{2u^2 + (2 \alpha p + 3) u + \alpha p + 1}}{u (u-\ell) (u-\ell+1)}.
\end{equation}
Summing this recursion relation over $\ell$ then yields \eqref{eq:Recurse}, upon noting that $(\alpha p + 1)^2 = 4t \Delta_p + 1$.  Since $f_{p;q}(0) = 0$ and $f_{p;q}(1) = 1$, it follows from \eqref{eq:Recurse} that $f_{p;q}(u)$ is a polynomial in $\lambda_{p;q}$ and $\Delta_p$, as claimed.
\end{proof}

Let $g_{u,v}(\lambda,\Delta)$ denote the polynomial for which $f_{p;q}(u) = g_{u,v}(\lambda_{p;q}, \Delta_p)$ and let
\begin{equation}
  I_{u,v}(\lambda,\Delta)=\prod_{(r,s)\in K(u,v)}\left(\Delta-\Delta_{r,s}\right)\cdot g_{u,v}(\lambda,\Delta).
\end{equation}
It is a simple corollary of \eqref{eq:SumFactorExs} and \eqref{eq:Recurse} that $g_{u,v}$ has degree $u-1$ as a polynomial in $\lambda$.  If we regard $\Delta$ as having degree $2$, then the total degree of $g_{u,v}$ is also $u-1$ and that of $I_{u,v}$ is therefore $(u-1)v$.
\begin{thm} \label{thm:ComputeZhu}
  Zhu's algebra of \(\AdmMod{u}{v}\) is given by the quotient
  \begin{align}
    \zhu{\AdmMod{u}{v}}=\frac{\UEA\:\SLA{sl}{2}}{\corrfn{I_{u,v}(h,T)}},
  \end{align}
  where $T = \frac{1}{2t} (\frac{1}{2} h^2 - ef - fe)$ denotes the image of $T(z)$ in $\zhu{\UVOA{k}} \cong \UEA\:\SLA{sl}{2}$.
\end{thm}
\begin{proof}
  As noted above, the ideal of $\UEA\:\SLA{sl}{2}$ by which one quotients to get the Zhu algebra $\zhu{\AdmMod{u}{v}}$ is generated by the image of the null field $\chi(z)$.  This image is a polynomial in \(e\), \(f\) and \(h\) of total degree at most \((u-1)v\). We have evaluated the action of the zero mode of \(\chi(z)\) on a continuum of \rhwvs{} and thus the image of \(\chi(z)\) in Zhu's algebra is, up to non-zero constant factors, equal to the polynomial \(I_{u,v}(h,T)\) of total degree $(u-1)v$.
\end{proof}

Before we can use this presentation of Zhu's algebra of \(\AdmMod{u}{v}\), we need to know a little more about the zeroes of the polynomial \(g_{u,v}\).
\begin{prop}\label{prop:zerosofg}
  For each $u \in \ZZ_{\ge 1}$, the polynomials $g_{u,v}(\lambda, \Delta_{r,0})$ evaluate to zero when $r = 1, 2, \ldots, u-1$ and $\lambda = r-1, r-3, \ldots, -r+3, -r+1$.
\end{prop}
\begin{proof}
This is trivial for $u=1$ as there are then no $r$ or $\lambda$ to check.  For $u=2$, \eqref{eq:SumFactorExs} gives $g_{2,v}(\lambda, \Delta) = \lambda$ and we need only check $r=1$ and $\lambda = 0$.  We may therefore assume, inductively, that the statement of the proposition is true for $g_{1,v}, g_{2,v}, \ldots, g_{u+1,v}$.  Then, the recursion relation \eqref{eq:Recurse} shows that
\begin{equation}
g_{u+2,v}(\lambda, \Delta_{r,0}) = \frac{(2u+1) \lambda}{(u+1)^2} g_{u+1,v}(\lambda, \Delta_{r,0}) - \frac{4t \Delta_{r,0} - (u-1)(u+1)}{(u+1)^2} g_{u,v}(\lambda, \Delta_{r,0})
\end{equation}
will vanish for all $r = 1, 2, \ldots, u-1$ and $\lambda = r-1, r-3, \ldots, -r+3, -r+1$, because $g_{u+1,v}$ and $g_{u,v}$ do.  Moreover, because $4t \Delta_{u,0} = (u-1)(u+1)$, $g_{u+2,v}$ also vanishes for $r=u$ and $\lambda = r-1, r-3, \ldots, -r+3, -r+1$.

The only remaining case is $r=u+1$.  Then, $4t \Delta_{u+1,0} = u(u+2)$, hence we may identify $\alpha p$ with $u$ or $-u-2$, hence $q = \frac{1}{2} (\lambda - \alpha p)$ with $\frac{1}{2} (\lambda - u)$ or $\frac{1}{2} (\lambda + u+2)$, respectively.  From \eqref{eq:DefF} and $\alpha p = u$, we now obtain
\begin{align}
g_{u+2,}(\lambda, \Delta_{u+1,0}) &= f_{u/\alpha; (\lambda - u)/2}(u+2) = \sum_{\ell=0}^{u+1} \binom{\frac{1}{2} (\lambda - u) - 1}{\ell} \binom{u+1+\ell}{\ell} \binom{2(u+1)}{u+1-\ell} \notag \\
&= \binom{2(u+1)}{u+1} \sum_{\ell=0}^{u+1} \binom{\frac{1}{2} (\lambda - u) - 1}{\ell} \binom{u+1}{u+1-\ell} = \binom{2(u+1)}{u+1} \binom{\frac{1}{2} (\lambda + u)}{u+1} \notag \\
&= \binom{2(u+1)}{u+1} \frac{(\lambda + u) (\lambda + u-2) \cdots (\lambda - u+2) (\lambda - u)}{2^{u+1} (u+1)!},
\end{align}
which clearly vanishes for $\lambda = u, u-2, \ldots, -u+2, -u$.  The result is the same for $\alpha p = -u-2$.
\end{proof}

We are now in a position to classify the simple weight modules over \(\zhu{\AdmMod{u}{v}}\).  Recall that \(\zhu{\AdmMod{u}{v}}\)-modules are automatically $\SLA{sl}{2}$-modules; by a weight module over \(\zhu{\AdmMod{u}{v}}\), we mean that it is a weight module over $\SLA{sl}{2}$.  The classification of simple \(\SLA{sl}{2}\) weight modules was summarised in \propref{prop:FinSL2WtMods}.  We also recall that the quadratic Casimir operator
\begin{equation}
Q=\frac{1}{2}h^2-ef-fe=2tT
\end{equation}
acts as a scalar multiple of the identity on any simple \(\SLA{sl}{2}\) weight module.

\begin{thm}\label{thm:ZhuClassification}
  The following $\SLA{sl}{2}$-modules provide a complete list of the inequivalent isomorphism classes of simple weight modules of \(\zhu{\AdmMod{u}{v}}\):
  \begin{itemize}
  \item The finite-dimensional highest and \lwms{} \(\FinIrr{\lambda_{r,0}}\), where \(1\le r\le u-1\).
  \item The infinite-dimensional \hwms{}
    \(\FinDisc{\lambda_{r,s}}\), where \(1\le r\le u-1\) and \(1\le s\le v-1\).
  \item The infinite-dimensional \lwms{} \(\finconjmod{\FinDisc{\lambda_{r,s}}}\), 
    where \(1\le r\le u-1\) and \(1\le s\le v-1\).
  \item The infinite-dimensional \rhwms{} \(\FinRel{\lambda,\Delta_{r,s}}\), where
    \((r,s)\in K(u,v)\) and \(4t\Delta_{r,s}\neq \mu(\mu+2)\) for all \(\mu\in \lambda+2\ZZ\).
  \end{itemize}
\end{thm}

\begin{proof}
  We first consider a simple finite-dimensional \(\zhu{\AdmMod{u}{v}}\)-module \(M\), which must therefore also   be a finite-dimensional \(\SLA{sl}{2}\)-module.  As the quadratic Casimir takes the value \(\frac{1}{2}(r^2-1)=2t \Delta_{r,0}\) on the \(r\)-dimensional simple \(\SLA{sl}{2}\) module, it follows that \(g_{u,v}(h,T)\) must act trivially on \(M\) because the remaining \(\lambda\)-independent factors of \(I_{u,v}(\lambda,\Delta)\) do not have \(\Delta_{r,0}\) as a root for any positive integer \(r\). By \propref{prop:zerosofg}, \(g_{u,v}(m,\Delta_{r,0})=0\) if \(1\le r\le u-1\) and \(m=r-1,r-3,\dots, -r+1\). Conversely if \(r\ge u\), then \(g_{u,v}(m,\Delta_{r,0})\neq0\) for some \(m=r-1,r-3,\dots, -1+r\), because the \(\lambda\)-degree of \(g_{u,v}(\lambda,\Delta_{r,0})\) is \(u-1\), so there cannot be more than \(u-1\) zeroes. Thus, \(M\) must be isomorphic to one of the \(\FinIrr{\lambda_{r,0}}\) for some \(1\le r\le u-1\).

  Next, we consider a simple infinite-dimensional \(\zhu{\AdmMod{u}{v}}\)-module \(M\), which must therefore also   be an infinite-dimensional \(\SLA{sl}{2}\) weight module. Because there must be an infinite number of weight vectors in \(M\) with distinct \(\SLA{sl}{2}\)-weights, \(g_{u,v}(h,T)\) cannot vanish identically on \(M\).  In order for \(I_{u,v}(h,T)\) to then vanish, \(T\) must act as multiplication by \(\Delta_{r,s}\) for some \((r,s)\in K(u,v)\). Referring to \propref{prop:FinSL2WtMods}, it follows that the last three cases of \thmref{thm:ZhuClassification} exhaust all the possible isomorphism classes for \(M\).
\end{proof}

\noindent The correspondence (\thmref{thm:ZhuSimples}) between simple modules of a \voa{} and its Zhu algebra then proves the following classification result:
\begin{thm} \label{thm:IrrRelMods}
  The following $\AKMA{sl}{2}$-modules provide a complete list of the inequivalent isomorphism classes of simple \rhwms{} of \(\AdmMod{u}{v}\):
  \begin{itemize}
  \item The \hwms{} \(\Irr{\lambda_{r,0}}\), where \(1\le r\le u-1\).
  \item The \hwms{} \(\Disc{\lambda_{r,s}}\), where \(1\le r\le u-1\) and \(1\le s\le v-1\).
  \item The conjugates \(\conjmod{\Disc{\lambda_{r,s}}}\), where \(1\le r\le u-1\) and \(1\le s\le v-1\).
  \item The \rhwms{} \(\Rel{\lambda,\Delta_{r,s}}\), where \((r,s)\in K(u,v)\) and \(4t\Delta_{r,s}\neq \mu(\mu+2)\) for all \(\mu\in \lambda+2\ZZ\).
  \end{itemize}
\end{thm}

\noindent These are therefore the simple modules of the \voa{} $\AdmMod{u}{v}$ which belong to the category $\categ{R}$ that was introduced in \secref{sec:Relaxed}.

Zhu's correspondence also extends to non-simple $\zhu{\AdmMod{u}{v}}$-modules.  In particular, $\SLA{sl}{2}$ admits reducible, but indecomposable, modules similar to the $\FinRel{\lambda; \Delta}$ of \propref{prop:FinSL2WtMods} whenever
\begin{equation} \label{eq:IndecCond}
4t \Delta = \mu(\mu+2) \quad \text{for some} \quad \mu \in \lambda+2\ZZ.
\end{equation}
For $\Delta = \Delta_{r,s}$, where $(r,s) \in K(u,v)$, the only solutions are $\mu = r-1-ts = \lambda_{r,s}$ and $\mu = -r-1+ts = \lambda_{u-r,v-s}$.  As $0<s<v$, we find that $\lambda_{r,s} - \lambda_{u-r,v-s} \notin 2 \ZZ$, concluding that an indecomposable with $\Delta = \Delta_{r,s}$ may have at most one weight $\mu$ satisfying \eqref{eq:IndecCond}.  Thus, there are \cite{MazLec10} precisely two reducible, but indecomposable, $\SLA{sl}{2}$-modules $\FinRel{\lambda_{r,s}; \Delta_{r,s}}^+$ and $\FinRel{\lambda_{r,s}; \Delta_{r,s}}^-$ for each $1 \le r \le u-1$ and $1 \le s \le v-1$.  They are determined (up to isomorphism) by the following non-split short exact sequences:
\begin{equation}
\dses{\FinDisc{\lambda_{r,s}}}{}{\FinRel{\lambda_{r,s},\Delta_{r,s}}^+}{}{\finconjmod{\FinDisc{\lambda_{u-r,v-s}}}}, \qquad 
\dses{\finconjmod{\FinDisc{\lambda_{u-r,v-s}}}}{}{\FinRel{\lambda_{r,s},\Delta_{r,s}}^-}{}{\FinDisc{\lambda_{r,s}}}.
\end{equation}
Applying Zhu's construction now leads to the following result:

\begin{thm} \label{thm:IndecRelMods}
For each $1 \le r \le u-1$ and $1 \le s \le v-1$, there exist two reducible, but indecomposable, $\AdmMod{u}{v}$-modules $\Rel{\lambda_{r,s}; \Delta_{r,s}}^+$ and $\Rel{\lambda_{r,s}; \Delta_{r,s}}^-$, obtained by inducing $\FinRel{\lambda_{r,s}; \Delta_{r,s}}^+$ and $\FinRel{\lambda_{r,s}; \Delta_{r,s}}^-$ to level $k$ $\AKMA{sl}{2}$-modules and quotienting by the sum of all the submodules that trivially intersect the space of conformal weight $\Delta_{r,s}$.  They are determined (up to isomorphism) by the following non-split short exact sequences:
\begin{equation}
\dses{\Disc{\lambda_{r,s}}}{}{\Rel{\lambda_{r,s},\Delta_{r,s}}^+}{}{\conjmod{\Disc{\lambda_{u-r,v-s}}}}, \qquad 
\dses{\conjmod{\Disc{\lambda_{u-r,v-s}}}}{}{\Rel{\lambda_{r,s},\Delta_{r,s}}^-}{}{\Disc{\lambda_{r,s}}}.
\end{equation}
\end{thm}

A theorem of Kac and Wakimoto \cite[Prop.~1]{KacMod88} asserts that the $\AdmMod{u}{v}$-modules in category $\categ{O}$ are all semisimple.  The category $\categ{O}$ $\AdmMod{u}{v}$-modules therefore consist of finite direct sums of the \hwms{} of \thmref{thm:IrrRelMods}.  By way of contrast, a corollary of \thmref{thm:IndecRelMods} is that the $\AdmMod{u}{v}$-modules of category $\categ{R}$ need not be semisimple (when $v \neq 1$).  We remark that we have not excluded the possibility that there exist \rhwms{} over $\AdmMod{u}{v}$ that extend the $\Rel{\lambda, \Delta_{r,s}}$, or the $\Rel{\lambda_{r,s}, \Delta_{r,s}}^{\pm}$, non-trivially; this would seem to require more information about the submodule structure of the relaxed Verma modules than is currently available, see \cite{FeiEqu98,SemEmb97}.  However, the \hw{} result given in \corref{cor:ClassHWMods} and the analogous results for the Virasoro minimal models suggest the following conjecture:

\begin{conj}
The $\AdmMod{u}{v}$-modules of \thmref{thm:IrrRelMods} and \thmref{thm:IndecRelMods} exhaust the indecomposable $\AdmMod{u}{v}$-modules of category $\categ{R}$.
\end{conj}

We close by demonstrating that the non-semisimplicity of \rhwms{} over $\AdmMod{u}{v}$ does not imply that the Virasoro mode $L_0$ acts non-semisimply, a fact that is of interest to \lcft{} studies.  This result requires the finite-dimensionality of the weight spaces of the category $\categ{R}$ modules, discussed in \secref{sec:Relaxed}.
\begin{thm} \label{thm:L0ActsSemisimply}
  \leavevmode
  \begin{enumerate}
  \item The image of the energy momentum tensor $T$ in \(\zhu{\AdmMod{u}{v}} \cong \UEA\:\SLA{sl}{2} / \corrfn{I_{u,v}(h,T)}\) acts semisimply on every weight module of \(\zhu{\AdmMod{u}{v}}\) with finite-dimensional weight spaces.
  \item The Virasoro zero mode \(L_0\) acts semisimply on every \rhwm{} of \(\AdmMod{u}{v}\).
  \end{enumerate}
\end{thm}

\begin{proof}
\leavevmode
  \begin{enumerate}
  \item Let \(\Mod{M}\) be an indecomposable weight module of \(\zhu{\AdmMod{u}{v}}\) on which (the image of) \(T\) acts non-semisimply. As $T$ is proportional to the image of the quadratic Casimir $Q$, it follows that $T$ has a single (generalised) eigenvalue on $\Mod{M}$. Then, by Weyl's theorem for $\SLA{sl}{2}$, \(\Mod{M}\) must be infinite-dimensional with an infinite number of distinct \(\SLA{sl}{2}\) weights.  Let \(\Mod{W}\) be the submodule of \(\Mod{M}\) spanned by the eigenvectors of \(T\). As $\Mod{W}$ is non-zero, $\Mod{W}$     must also be infinite-dimensional as it possesses an infinite number of distinct \(\SLA{sl}{2}\)-weights.  Thus, \(\Mod{W}\) is an eigenspace of \(T\) with eigenvalue \(\Delta_{r,s}\), for some \((r,s)\in K(u,v)\).
    
  Next, assume that there exists a generalised eigenvector \(v\) of \(T\), so that $(T-\Delta_{r,s})v\neq0$, with \(\SLA{sl}{2}\)-weight \(\lambda_v\).  The existence of \(v\) would imply that \(I_{u,v}(\lambda_v, \Delta)\) has a zero of order at least $2$ at \(\Delta=\Delta_{r,s}\). Since there are only finitely many \(\SLA{sl}{2}\)-weights \(\lambda\) for which \(\Delta=\Delta_{r,s}\) is a zero of \(I_{u,v}(\lambda, \Delta)\) of order at least $2$, the quotient \(\Mod{M}/\Mod{W}\) must be finite-dimensional. But, the eigenvalue of $T$ on $v+\Mod{W}$ is then \(\Delta_{r,0}\), for some \(1\le r\le u-1\), which is a contradiction.  It follows that no such generalised eigenvectors exist, hence that $T$ acts semisimply on $\Mod{M}$.
  \item On any \rhwm{} over $\AdmMod{u}{v}$, the action of $L_0$ on the \rhwvs{} coincides with that of $T$ on the corresponding $\zhu{\AdmMod{u}{v}}$-module.  As the latter action is semisimple, so is that of $L_0$ on the \rhwvs{}.  As these generate the whole module, $L_0$ acts semisimply. \qedhere
  \end{enumerate}
\end{proof}

We stress that this result does not imply that the \cfts{} corresponding to the \voas{} $\AdmMod{u}{v}$ are non-logarithmic.  Indeed, it has been known for some time that there are models with $k=-\tfrac{4}{3}$ \cite{GabFus01} and $k=-\tfrac{1}{2}$ \cite{LesLog04,RidFus10} that are logarithmic.  The loophole is that we may also twist by the so-called spectral flow automorphisms which, when $v \neq 1$, lead to infinitely many new simple (and indecomposable) $\AdmMod{u}{v}$-modules that do not belong to category $\categ{R}$.  For $k=-\tfrac{4}{3}$ and $k=-\tfrac{1}{2}$, there exist indecomposable $\AdmMod{u}{v}$-modules that are formed from \rhwms{} from different spectral flow sectors and the action of $L_0$ on these is non-semisimple.  In fact, these modules are \emph{staggered} in the sense of \cite{RidSta09,CreLog13}; a detailed discussion may be found in \cite{CreMod12}.  We expect that there exist staggered $\AdmMod{u}{v}$-modules whenever $v \neq 1$, hence that the associated \cfts{} are logarithmic, and hope to report on this in the future.

\section*{Acknowledgements}

We thank Jim Borger for help with a question of commutative algebra, J\"{u}rgen Fuchs, Masoud Kamgarpour and Christoph Schweigert for illuminating discussions regarding parabolic Verma modules, Antun Milas for correspondence concerning the current status of higher rank generalisations, Ole Warnaar for advice on symmetric function theory, and the organisers of the Erwin Schr\"{o}dinger Institute programme ``Modern trends in topological quantum field theory'' for their hospitality.
DR's research is supported by the Australian Research Council Discovery Project DP1093910.  
SW's work is supported by the Australian Research Council Discovery Early Career Researcher Award DE140101825.

\appendix
\section{Symmetric Polynomials} \label{sec:SymmPoly}

The standard reference for the parts of symmetric function theory most applicable to the work reported here is Macdonald's book \cite{MacSym95}.  In this appendix, we summarise the results from Chapters 1 and 6 of this book that we use freely throughout.

\subsection{Partitions of integers}\label{sec:Partitions}

A number of bases of the ring of symmetric polynomials are indexed by partitions. We therefore fix some notation on partitions before going on to discuss symmetric polynomials. A partition \(\lambda=[\lambda_1,\dots,\lambda_m]\) is a weakly descending sequence of positive integers called parts. The length \(\ell(\lambda)=m\) is the length of the sequence and the weight \(\abs{\lambda}=\sum_i\lambda_i\) is the sum over all elements of the sequence. Sometimes, it is convenient to define the \(\lambda_i\) for \(i>\ell(\lambda)\) to be $0$. A partition \(\lambda\) is often also referred to as a partition of the integer \(\abs{\lambda}\); $[3,3,2,1,1,1]$ is thus a partition of $11$.  It is customary to regard the empty partition $[\:]$ as a partition of $0$.

A convenient shorthand for partitions is to indicate repeated parts using a superscript.  Thus, $[3,3,2,1,1,1]$ and $[3^2,2,1^3]$ denote the same partition.  The multiplicity in $\lambda$ of a given part \(i\), that is, the superscript in the convenient shorthand notation, will be denoted by \(m_\lambda(i)\). For every partition \(\lambda\), one may then introduce the following number:
\begin{equation}
  z_\lambda=\prod_{i\ge 1}m_\lambda(i)!\cdot i^{m_\lambda(i)}.
\end{equation}
These numbers play a small role in what follows, see \eqref{eq:DefInner} below for example.

One also associates, to each partition \(\lambda\), a diagram called a Young diagram. This consists of \(\ell(\lambda)\) rows of left-aligned boxes for which the length of the \(i\)-th row is \(\lambda_i\). We draw the first row at the top and the \(\ell(\lambda)\)-th at the bottom. With this convention, the conjugate partition \(\lambda^\prime\) is defined to be the partition whose Young diagram is the reflection of that of $\lambda$ along the diagonal line from top left to bottom right. This reflection exchanges the lengths of the columns and the rows.

Each box $b$ of the Young diagram of $\lambda$, generally written as \(b\in\lambda\), may be parametrised by a pair \(b=(i,j)\), where \(i\) is the row number counted from top to bottom and \(j\) is the column number counted from left to right. Given a partition \(\lambda\) and a box \(b\in\lambda\), the arm length \(a(b)\), the leg length \(l(b)\), the arm colength \(a^\prime(b)\) and the leg colength \(l^\prime(b)\) are the distances from \(b\) to the right, bottom, left and top edges of the Young diagram, respectively. In formulae,
\begin{equation} \label{eq:ArmLeg}
a(b)=\lambda_i-j,\qquad
a^\prime(b)=j-1,\qquad
l(b)=\lambda^\prime_j-i,\qquad
l^\prime(b)=i-1.
\end{equation}

Finally, we remark that partitions admit a number of partial orderings, among which the dominance ordering, which we denote by $\ge$, will prove useful. Two partitions \(\lambda\) and \(\mu\) satisfy \(\lambda\ge \mu\) if \(\abs{\lambda}=\abs{\mu}\) and
\begin{align}
  \sum_{i=1}^m \lambda_i \ge \sum_{i=1}^m \mu_i,\quad \text{for all \(m\ge 1\).}
\end{align}

\subsection{Symmetric polynomials}\label{sec:SubSymmPoly}

Let \(\sympol{n}\) be the ring of $n$-variable polynomials with complex coefficients that are invariant under arbitrary permutations of the variables. $\sympol{n}$ is called the ring of symmetric polynomials in \(n\) variables. This ring is graded,
\begin{equation}
  \sympol{n}=\bigoplus_{k\ge 0}\sympol{n}^k,
\end{equation}
where \(\sympol{n}^k\) is the space of homogeneous symmetric polynomials of degree \(k\).

Examples of symmetric polynomials include:
\begin{enumerate}
\item \emph{Power sums}: The symmetric polynomial
  \begin{equation}
    \fpowsum{k}{x_1,\dots,x_n}=\sum_{i=1}^n x_i^k, \qquad k\ge 1,
  \end{equation}
  is called the \(k\)-th power sum. For each partition \(\lambda=(\lambda_1,\dots,\lambda_m)\), we define
  \begin{equation}
    \powsum{\lambda}=\powsum{\lambda_1}\cdots \powsum{\lambda_m}.
  \end{equation}
\item \emph{Elementary symmetric polynomials}: The symmetric polynomials
  \begin{equation}
      \felsym{0}{x_1,\dots,x_n}=1, \qquad 
      \felsym{k}{x_1,\dots,x_n}=\sum_{1\le i_1<\cdots<i_k\le n} x_{i_1}\cdots x_{i_k}, \quad 1\le k\le n,
  \end{equation}
  are called the elementary symmetric polynomials. The elementary symmetric polynomials appear in the expansion of the generating function
  \begin{equation}
    \prod_{i=1}^n(y+x_i)=\sum_{i=0}^n \felsym{i}{x_1,\dots,x_n} y^{n-i}.
  \end{equation}
  \item \emph{Monomial symmetric polynomials}: Let
  \(\alpha=(\alpha_1,\dots,\alpha_n)\in \mathbb{Z}_{\ge 0}^n\) be an \(n\)-tuple of non-negative integers.  The symmetric polynomials
  \begin{equation}
    \fmonsym{\alpha}{x_1,\dots,x_n}=\sum_{\sigma}x_1^{\sigma_1}\cdots x_n^{\sigma_n},
  \end{equation}
  where the sum runs over all distinct permutations \(\sigma\) of \(\alpha\),
  are called the monomial symmetric polynomials or symmetric monomials for short. Given \(\alpha\), there is always precisely one distinct permutation $\sigma$ whose entries are in descending order, hence by omitting any trailing zeroes, we may assume that \(\alpha\) is a partition of length at most \(n\). 
\end{enumerate}

\begin{prop}
  \leavevmode
  \begin{enumerate}
  \item The power sums $\powsum{k}$ are algebraically independent for \(k\le n\) and generate \(\sympol{n}\):
    \begin{equation}
      \sympol{n} =\mathbb{C}[\powsum{1},\dots,\powsum{n}].
    \end{equation}
  \item The elementary symmetric polynomials are algebraically independent and generate \(\sympol{n}\):
    \begin{equation}
      \sympol{n} =\mathbb{C}[\elsym{1},\dots,\elsym{n}].
    \end{equation}
  \item The symmetric monomials \(\monsym{\lambda}\) with $\ell(\lambda) \le n$ form a basis of \(\sympol{n}\).
  \end{enumerate}
\end{prop}

When working with symmetric polynomials, it is a remarkable fact that the number of variables is often irrelevant, assuming only that this number is sufficiently large. For this reason, it is convenient to work with symmetric polynomials in infinitely many variables. The ring of symmetric polynomials in infinitely many variables, $x_i$ say, is given by an inverse limit:
\begin{equation}
  \symfunc=\varprojlim_{n}\sympol{n}.
\end{equation}
Often the elements of $\symfunc$ are distinguished by referring to them as symmetric functions.  The projection
\begin{equation}
  \pi_n\colon\symfunc \rightarrow \sympol{n},
\end{equation}
defined by setting \(x_i=0\) for all \(i>n\), recovers the case of finitely many variables.

\begin{prop}
  \leavevmode
  \begin{enumerate}
  \item The power sums are algebraically independent for all \(k\ge1\) and generate \(\symfunc\):
    \begin{equation}
      \symfunc =\mathbb{C}[\powsum{1},\powsum{2},\dots].
    \end{equation}
    In addition, the \(\powsum{\lambda}\) form a basis of \(\symfunc\) as \(\lambda\) runs over all partitions of all non-negative integers.
  \item The symmetric monomials \(\monsym{\lambda}\) form a basis of \(\symfunc\) as \(\lambda\) runs over all partitions of all non-negative integers. The projection \(\pi_n\colon\symfunc\rightarrow \sympol{n}\) satisfies \(\pi_n(\monsym{\lambda})=0\) if and only if \(\ell(\lambda)>n\).
  \end{enumerate}
\end{prop}

The basis of \(\symfunc\) which is most interesting for the purposes of this article is that consisting of the Jack symmetric polynomials. These polynomials are orthogonal with respect to the following inner product on \(\symfunc\):
\begin{align} \label{eq:DefInner}
  \inprod{\powsum{\lambda}}{\powsum{\mu}}_t=z_\lambda t^{\ell(\lambda)}\delta_{\lambda,\mu}.
\end{align}
In fact, there are infinitely many families of Jack polynomials, each forming a basis of $\symfunc$, parametrised by the complex number \(t\in\mathbb{C}\setminus\mathbb{Q}_{\le 0}\) appearing in this inner product.
\begin{prop}
  \leavevmode
  \begin{enumerate}
  \item Let \(\{u_m\}\) be a basis of \(\symfunc\) and let \(\{v_m\}\) be the dual basis with respect to the inner product \eqref{eq:DefInner}. Then, the following identities hold:
    \begin{equation} \label{eq:Cauchy}
      \prod_{k\ge1} \exp \left( \frac{1}{t} \frac{\fpowsum{k}{x_1,x_2,\dots} \fpowsum{k}{y_1,y_2,\dots}}{k} \right)
      = \prod_{i,j\ge 1} (1-x_i y_j)^{-1/t}
      = \sum_m u_m(x_1,x_2,\dots) v_m(y_1,y_2,\dots).
    \end{equation}
  Analogous identities hold for finite numbers of variables $x_i$ and/or $y_j$ by projection.
  \item For every partition \(\lambda\) and every $t \in \CC \setminus \QQ_{\le 0}$, there exists a unique basis of symmetric polynomials \(\jack{\lambda}{t}\) such that $\inprod{\jack{\lambda}{t}}{\jack{\mu}{t}}_t=0$, whenever $\lambda\neq\mu$, and such that they satisfy upper-triangular decompositions into symmetric monomials,
  \begin{equation} \label{eq:JackMonomials}
    \jack{\lambda}{t}=\sum_{\mu\le\lambda} u_{\lambda,\mu}(t) \monsym{\mu},
  \end{equation}
  where \(u_{\lambda,\mu}(t)\in\mathbb{C}\) and \(u_{\lambda,\lambda}(t)=1\).  Here, the sum runs over all \(\mu\) that are dominated by $\lambda$.  It follows that the $\jack{\lambda}{t}$ are homogeneous symmetric polynomials of degree $\abs{\lambda}$.
  \end{enumerate}
\end{prop}
\begin{defn}
The \(\jack{\lambda}{t}\) determined by the previous proposition are called the \emph{Jack symmetric polynomials} or, when working in the infinite-variable ring $\symfunc$, the \emph{Jack symmetric functions}.
\end{defn}
\noindent We remark that if \(t\in\mathbb{Q}_{\le 0}\), then some of the coefficients \(u_{\lambda,\mu}(t)\) in \eqref{eq:JackMonomials} will diverge for certain \(\lambda\), hence some of the Jack polynomials will not be defined.

The Jack polynomials \(\jack{\lambda}{t}\) satisfy a number of remarkable properties. Before we can list those that we shall require, we need to introduce one more inner product, this time that of two symmetric polynomials \(f,g \in \sympol{n}\) of $n$ variables:
\begin{equation}\label{eq:innerprod}
  \intprod{f}{g}_t^n=\int_{[\Delta_n]}\prod_{1\le i\neq j\le n}\brac{1-\frac{x_i}{x_j}}^{1/t} \overline{f(x_1,\dots,x_n)} g(x_1,\dots,x_n)\: \frac{\dd x_1\cdots \dd x_n}{x_1\cdots x_n}.
\end{equation}
Here, the overline indicates that the arguments of the function have been inverted:  $\overline{f(x_1,\dots,x_n)} = f(x_1^{-1},\dots,x_n^{-1})$.  We remark that the integral in \eqref{eq:innerprod} may be thought of as an $n$-variable generalisation of taking the residue at $0$ of a meromorphic function.  In particular, the integral vanishes if $f$ and $g$ are homogeneous of different degrees.

In \cite{MacSym95}, an inner product is defined for a deformation of the Jack polynomials (now) called the Macdonald polynomials. The definition is almost identical to \eqref{eq:innerprod}, utilising a cycle $[\Delta_n]$ that is just a (normalised) product of \(n\) unit circles. However, there are some subtle ties when passing to the Jack polynomial limit and, in particular, this cycle is no longer suitable. Instead, we will use the (normalised) cycles constructed by Tsuchiya and Kanie which also have the advantage of being supported in the domains required for radially-ordered expansions of screening operators.
\begin{cthm}[\protect{\cite{TsuFoc86}}]\label{thm:tkthm}
Let \(r\in\mathbb{Z}_{\ge 1}\) and \(t\in\mathbb{C}^\ast\) and suppose that \(d(d+1)/t\notin\mathbb{Z}\) and \(d(r-d)/t \notin \mathbb{Z}\), for all integers \(d\) satisfying \(1\le d\le r-1\).  Then, there exists a cycle \(\Delta_r\) such that for each symmetric Laurent polynomial \(f(z_1,\dots,z_r)\), the integral
\begin{equation} \label{eq:Contour+}
\int_{\Delta_r} G_r(z;t) f(z_1,\dots,z_r)\: \frac{\dd z_1\cdots\dd z_r}{z_1\cdots z_r}
\end{equation}
is equal to
\begin{equation}
\int_{\abs{z}=1}\int_{\sigma_{r-1}}\prod_{i=1}^{r-1}(1-y_i)^{2/t}\cdot G_{r-1}(y;t) f(z,zy_1,zy_2,\dots,zy_{r-1}) \: \frac{\dd y_1 \cdots \dd y_{r-1} \: \dd z}{y_1 \cdots y_{r-1} z},
\end{equation}
where \(\sigma_{r-1}\) is (a regularisation of) the \((r-1)\)-simplex \(\{1>y_1>\cdots>y_{r-1}>0\}\).  If \(r=1\), then \(G_1(z;t)=1\) and \(\Delta_1\) is just the unit circle.  In particular, if \(f(z_1,\dots,z_r)=1\), then \eqref{eq:Contour+} may be evaluated as a Selberg integral:
\begin{equation} \label{eq:Selberg}
  S_n(t)=\int_{\Delta_r} G_r(z;t) \: \frac{\dd z_1\cdots\dd z_r}{z_1\cdots z_r} = \frac{2\pi \ii}{(r-1)!} \prod_{j=1}^{r-1}\frac{\Gamma(1+(j+1)/t)\Gamma(-j/t)}{\Gamma(1+1/t)}.
\end{equation}
This integral is non-zero, hence the cycle $\Delta_r$ is non-trivial.
\end{cthm}
\noindent The normalised cycle $[\Delta_n]$ is then given by \([\Delta_n]=\Delta_n/S_n(t)\), so that $\intprod{1}{1}^n_t=1$.

We finish by summarising the properties of the Jack symmetric polynomials that we will use in this article.
\begin{prop} \label{prop:jackprops}
  \leavevmode
  \begin{enumerate}
  \item The elementary symmetric polynomials are Jack polynomials (for all values of \(t\)):  $\elsym{k} = \jack{[1^k]}{t}$.
  \item The norm squared of \(\jack{\lambda}{t}\) with respect to the infinite variable inner product \(\inprod{\cdot}{\cdot}_t\) is
    \begin{align}
      \inprod{\jack{\lambda}{t}}{\jack{\lambda}{t}}_t
      =\prod_{b\in\lambda} \frac{t(a(b)+1)+l(b)}{t a(b)+l(b)+1}.
    \end{align}
    The corresponding dual basis of the Jack polynomials will be denoted by
    \begin{align} \label{eq:DefQ}
      \djack{\lambda}{t} = \frac{\jack{\lambda}{t}}{\inprod{\jack{\lambda}{t}}{\jack{\lambda}{t}}_t}.
    \end{align}
  \item\label{item:specialisation} For any \(X\in\mathbb{C}\), one defines a ring homomorphism $\Xi_X \colon \symfunc \to \mathbb{C}$, called the \emph{specialisation map}, by $\Xi_X(\powsum{k})=X$, for all $k\ge1$. The Jack polynomials and their duals
    specialise to
    \begin{equation}
      \Xi_X(\jack{\lambda}{t})=\prod_{b\in\lambda}\frac{X+ta^\prime(b)-l^\prime(b)}{ta(b)+l(b)+1}, \qquad
      \Xi_X(\djack{\lambda}{t})=\prod_{b\in\lambda}\frac{X+t a^\prime(b)-l^\prime(b)}{t(a(b)+1)+l(b)}.
    \end{equation}
  \item The projection onto \(n\) variables satisfies \(\pi_n(\jack{\lambda}{t})=0\) if and only if \(\ell(\lambda)>n\).
  \item In \(\sympol{n}\), one has $\fjack{[m^n]}{t}{x_1, \ldots, x_n} = \fmonsym{[m^n]}{x_1, \ldots, x_n} = \prod_{i=1}^n x_i^m$. \label{item:jackmonomials}
  \item Suppose that \(\lambda\) satisfies \(\ell(\lambda) \le n\) and let \(\lambda+[m^n]\) denote the partition \([\lambda_1+m,\dots,\lambda_n+m]\). Then, in \(\sympol{n}\),
    \begin{align} \label{eq:JackByRectJack}
      \prod_{i=1}^n x_i^m \cdot \jack{\lambda}{t}=\jack{\lambda+[m^n]}{t}.
    \end{align}
  \item Suppose that \(\lambda\) satisfies \(\ell(\lambda) \le n\).  Then, the norm squared of the \(\jack{\lambda}{t}\) with respect to the finite variable inner product \(\intprod{\cdot}{\cdot}^n_t\) is
    \begin{align}
      \intprod{\jack{\lambda}{t}}{\jack{\lambda}{t}}^n_t=\prod_{b\in\lambda}
      \frac{\tbrac{t(a(b)+1)+l(b)}\tbrac{n+ta^{\prime}(b)-l^\prime(b)}}
      {\tbrac{ta(b)+l(b)+1}\tbrac{n+t(a^{\prime}(b)+1)-l^\prime(b)-1}}.
    \end{align}
  \end{enumerate}
\end{prop}

The specialisation map of \propref{prop:jackprops}, item \ref{item:specialisation}, seems somewhat mysterious at first glance and deserves some additional explanation. Our chief use for it is to expand products of the form \(\prod_{i\ge1}(1-x_i)^{-X/t}\) in terms of Jack polynomials.  To do this, we consider \eqref{eq:Cauchy} with $y_1 = 1$ and $y_i = 0$ for all $i>1$:
\begin{equation}
\prod_{i\ge1} \brac{1-x_i}^{-1/t} = \prod_{k\ge1} \exp \brac{\frac{1}{t} \frac{\fpowsum{k}{x_1, \ldots, x_n}}{k}}.
\end{equation}
It follows that we may now write
\begin{align} \label{eq:SpecialisingToJacks}
\prod_{i\ge1}(1-x_i)^{-X/t} &= \prod_{k\ge1} \exp \brac{\frac{X}{t} \frac{\fpowsum{k}{x_1, x_2, \ldots}}{k}} = \func{\Xi_X}{\prod_{k\ge1} \exp \brac{\frac{1}{t} \frac{\fpowsum{k}{x_1, x_2, \ldots} \fpowsum{k}{y_1, y_2, \ldots}}{k}}} \notag \\
&= \func{\Xi_X}{\sum_{\lambda} \fjack{\lambda}{t}{x_1, x_2, \ldots} \fdjack{\lambda}{t}{y_1, y_2, \ldots}} = \sum_{\lambda} \fjack{\lambda}{t}{x_1, x_2, \ldots} \func{\Xi_X}{\fdjack{\lambda}{t}{y_1, y_2, \ldots}} \notag \\
&= \sum_{\lambda} \sqbrac{\prod_{b\in\lambda}\frac{X+t a^\prime(b)-l^\prime(b)}{t(a(b)+1)+l(b)}} \fjack{\lambda}{t}{x_1, x_2, \ldots},
\end{align}
where we take the specialisation map $\Xi_X$ to act only on the symmetric polynomials in the \(y_i\).

\section{Zhu's algebra} \label{sec:zhu}

From any \voa{} \(\VOA{V}\), one can construct a unital associative algebra \(\zhu{\VOA{V}}\) called Zhu's algebra. The representation theory of this associative algebra is closely related that of \(\VOA{V}\) and is a crucial tool for classifying \(\VOA{V}\)-modules, especially simple \(\VOA{V}\)-modules.  In this appendix, we motivate Zhu's algebra by using generalised commutation rules to relate it to the algebra of zero modes of the fields of $\VOA{V}$.  From this point of view, the technology of Zhu may be regarded as an abstract formalisation of the ``annihilating ideals'' discussed in the physics literature \cite{FeiAnn92}.  A part of this motivational discussion may be found in \cite{KacBom88}; another means of motivating Zhu's algebra via deforming the standard conventions for normal ordering is the subject of \cite{BruAss99}.

A weight module $\Mod{M}$ of \(\VOA{V}\) is said to be \(\NN\)-graded if its decomposition into generalised \(L_0\)-eigenspaces is bounded below, that is, if there exists an \(h\in\mathbb{C}\) so that
\begin{equation}
  \Mod{M}=\bigoplus_{n\ge0} \Mod{M}_n,\quad \Mod{M}_n=\{u\in \Mod{M}\st (L_0-h-n)^mu=0 \text{ for some } m\in\ZZ_{\ge 1}\}.
\end{equation}
The space \(\overline{\Mod{M}}=\Mod{M}_0\) is often called the space of ground states; these states are clearly annihilated by every positive mode of every field.  We remark that in the language of \secref{sec:Relaxed}, the space $\overline{\Mod{M}}$ is spanned by \rhwvs{} (though $\Mod{M}$ may contain other such vectors that are not in $\overline{\Mod{M}}$).

Our preferred means to motivate the construction of Zhu's algebra is the observation that a simple weight module $\Mod{M}$ of $\VOA{V}$ is, rather generally, completely determined by its space of ground states $\overline{\Mod{M}}$.  This space admits an action of the zero modes of the fields and Zhu's algebra is an abstract realisation of this action.  To be more precise, consider \cite{KacBom88} the generalised commutation relation obtained from the contour integral
\begin{equation}
\oint_0 \oint_w \frac{A(z) B(w) \: z^{h_A} w^{h_B-1}}{z-w} \: \frac{\dd z}{2 \pi \ii} \frac{\dd w}{2 \pi \ii},
\end{equation}
where $A(z)$ and $B(z)$ are two fields of $\VOA{V}$ of conformal weights $h_A$ and $h_B$, respectively.  (We denote their modes by $A_j$ and $B_j$, respectively, and the corresponding states by $A$ and $B$.)  Without the denominator in the integrand, the usual procedure of breaking the inner contour in two or employing the \ope{} would lead to the commutation rules of the modes $A_1$ and $B_0$.  With the denominator, the generalised commutation relation is
\begin{equation}
\sum_{j\ge 0} \sqbrac{A_{-j} B_j + B_{-j-1} A_{j+1}} = \sum_{j\ge -h_A} \binom{h_A}{j+h_A} (A_jB)_0,
\end{equation}
where $(A_jB)_0$ denotes the zero mode of the field corresponding to the state $A_j B$.  Letting this generalised commutation relation act on a ground state $v \in \overline{\Mod{M}}$ (or a \rhwv{}), we arrive at
\begin{equation} \label{eq:GCR1}
A_0 B_0 v = \sum_{j\ge -h_A} \binom{h_A}{j+h_A} (A_jB)_0 v.
\end{equation}

This motivates the definition \cite{ZhuMod96} of Zhu's product $\ast$ on the \voa{} $\VOA{V}$ (we identify its elements with the states of the vacuum module for convenience):
\begin{equation} \label{eq:Zhu1}
A \ast B = \sum_{j\ge -h_A} \binom{h_A}{j+h_A} A_jB = \oint_0 A(z) B \: \frac{(1+z)^{h_A}}{z} \: \frac{\dd z}{2 \pi \ii}.
\end{equation}
The vacuum $\Omega$ is easily verified to be a (two-sided) unit with respect to this operation.  To illustrate, we tabulate the products of the generators $e$, $h$ and $f$ of the level $k$ universal \voa{} $\UVOA{k}$ of $\AKMA{sl}{2}$:
\begin{equation}
\begin{tabular}{C|CCC}
\ast & e & h & f \\
\hline
e & \normord{ee} & \normord{eh}-2e & \normord{ef}-h \\
h & \normord{he}+2e & \normord{hh} & \normord{hf}-2f \\
f & \normord{fe}+h & \normord{fh}+2f & \normord{ff}
\end{tabular}
.
\end{equation}
For example, $h \ast e = h_{-1} e + h_0 e = \normord{he} + 2e$.  We see immediately that, unlike the algebra of zero modes on the ground states, this $\ast$ operation does not respect the $\SLA{sl}{2}$ commutation rules, for example, $h \ast e - e \ast h = 4e + 2 \pd e$.  Moreover, it is not even associative, for example, $(h \ast e) \ast e - h \ast (e \ast e) = 2 \normord{\pd e e} + 2 \normord{ee}$.

Consider now, for each $\NN$-graded $\VOA{V}$-module $\Mod{M}$, the map taking a \voa{} element $A \in \VOA{V}$ to the linear endomorphism $A_0$, restricted to the space of ground states $\overline{\Mod{M}}$.  From \eqref{eq:GCR1} and \eqref{eq:Zhu1}, we see that
\begin{equation}
\pi_{\Mod{M}}(A \ast B) = \pi_{\Mod{M}}(A) \pi_{\Mod{M}}(B).
\end{equation}
Because the algebra of zero modes is associative, it follows that the failure of $\ast$ to be associative on $\VOA{V}$ must be explained by the discrepancies being mapped to zero:
\begin{equation}
\pi_{\Mod{M}} \bigl( (A \ast B) \ast C - A \ast (B \ast C) \bigr) = 0,
\end{equation}
for all \(A, B, C \in \VOA{V}\) and all $\NN$-graded $\VOA{V}$-modules $\Mod{M}$.  Similarly, discrepancies between the commutation rules of the zero modes and the $\ast$-commutation rules must also map to zero.  In accord with these observations, it is easy to check that $\pi_{\Mod{M}}$ does map $h*e - e*h - 2e = 2(e + \pd e)$  and $(h \ast e) \ast e - h \ast (e \ast e) = 2 (\normord{\pd e e} + \normord{ee})$ to zero.

To accurately reflect the algebra of zero modes acting on the ground states, we should therefore quotient the \voa{} by the intersection (over all $\Mod{M}$) of the kernels of the $\pi_{\Mod{M}}$.  One can obtain many elements of this kernel by again appealing to generalised commutation relations, in this case that obtained from 
\begin{equation}
\oint_0 \oint_w \frac{A(z) B(w) \: z^{h_A+1} w^{h_B-1}}{(z-w)^2} \: \frac{\dd z}{2 \pi \ii} \frac{\dd w}{2 \pi \ii}.
\end{equation}
The resulting generalised commutation relation, applied to $v \in \overline{\Mod{M}}$, takes the form
\begin{equation} \label{eq:GCR2}
A_0 B_0 v = \sum_{j\ge -h_A-1} \binom{h_A+1}{j+h_A+1} (A_jB)_0 v \qquad \Ra \qquad \sum_{j\ge -h_A-1} \binom{h_A}{j+h_A+1} (A_jB)_0 v = 0,
\end{equation}
where we have combined this relation with that of \eqref{eq:GCR1} in order to obtain the vanishing condition.

This motivates the definition \cite{ZhuMod96} of Zhu's other product $\circ$ on $\VOA{V}$:
\begin{equation} \label{eq:Zhu2}
A \circ B = \sum_{j\ge -h_A-1} \binom{h_A}{j+h_A+1} A_jB = \oint_0 A(z) B \: \frac{(1+z)^{h_A}}{z^2} \: \frac{\dd z}{2 \pi \ii}.
\end{equation}
It is clear from \eqref{eq:GCR2} and \eqref{eq:Zhu2} that all elements of the form $A \circ B$ belong to the kernel of every $\pi_{\Mod{M}}$.  Indeed, Zhu showed that the space $O(\VOA{V})$ spanned by the elements of this form is a two-sided ideal of $\VOA{V}$.  There appears to be some confusion in the literature as to whether Zhu proved explicitly that $O(\VOA{V})$ is in fact the intersection of the kernels of the $\pi_{\Mod{M}}$ (for example, a remark amounting to this is stated without proof in \cite{FreVer92}), but this result may be found in \cite[App.~A.2]{MatZer05}.  In any case, the quotient $\zhu{\VOA{V}} = \VOA{V} / O(\VOA{V})$ is what is now referred to as Zhu's algebra.  It is a unital associative algebra with respect to Zhu's $\ast$ product.

Consideration of the elements of $O(\VOA{V})$, for example, $T + \pd T = T \circ \Omega \in O(\VOA{V})$, shows that Zhu's algebra is not graded by conformal weight.  It is, however, filtered by conformal weight in the sense that it has an increasing sequence of subspaces $\filtzhu{0}{\VOA{V}} \subseteq \filtzhu{1}{\VOA{V}} \subseteq \filtzhu{2}{\VOA{V}} \subseteq \cdots$, where $\filtzhu{m}{\VOA{V}}$ is the image in $\zhu{\VOA{V}}$ of $\bigoplus_{n=0}^m \VOA{V}_n$ (and $\VOA{V}$ is here regarded as an $\NN$-graded module over itself).

It should now be clear that the space \(\overline{\Mod{M}}\) of ground states of a \(\VOA{V}\)-module \(\Mod{M}\) is an \(\zhu{\VOA{V}}\)-module. The converse to this statement is of great importance to the representation theory of \voas{}. 
\begin{cthm}[\protect{\cite[Thms.~2.2.1--2]{ZhuMod96}}] \label{thm:ZhuSimples}
  There is a bijective correspondence between isomorphism classes of simple \(\zhu{\VOA{V}}\)-modules \(\overline{\Mod{M}}\) and simple \(\NN\)-graded \(\VOA{V}\)-modules. More precisely, for every simple \(\NN\)-graded \(\VOA{V}\)-module \(\Mod{M}\), the ground states form a simple \(\zhu{\VOA{V}}\)-module $\overline{\Mod{M}}$ and for every simple \(\zhu{\VOA{V}}\)-module \(\overline{\Mod{M}}\), there exists a simple \(\NN\)-graded \(\VOA{V}\) module $\Mod{M}$ with \(\overline{\Mod{M}}\) as its space of ground states. 
\end{cthm}

Zhu actually gives a construction of a universal \(\NN\)-graded \(\VOA{V}\)-module \(\Mod{M}\) from a simple \(\zhu{\VOA{V}}\)-module \(\overline{\Mod{M}}\) so that any \(\NN\)-graded \(\VOA{V}\)-module \(\Mod{M}^\prime\) with \(\overline{\Mod{M}}\) as its space of ground states is a quotient of \(\Mod{M}\). If the \voa{} \(\VOA{V}\) is rational, then its modules are semisimple and thus \(\Mod{M}\) is simple. However, for more general \voas{} the universal \(\Mod{M}\) that Zhu constructs may be reducible.  Moreover, one may also start with a non-simple $\zhu{\VOA{V}}$-module; the result is then always reducible. 

We close this appendix with two results of Frenkel and Zhu that will be used in this paper.
\begin{cprop}[\protect{\cite[Prop.~1.4.2]{FreVer92}}] \label{prop:zhuideal}
  Let \(\VOA{V}\) be a \voa{} with an ideal \(\VOA{I}\) that does not contain the vacuum $\Omega$ or the conformal vector $T$. Then, the image $\zhu{\VOA{I}}$ of $\VOA{I}$ in $\zhu{\VOA{V}}$ is a two-sided ideal satisfying 
  \begin{align}
    \zhu{\VOA{V}/\VOA{I}}=\frac{\zhu{\VOA{V}}}{\zhu{\VOA{I}}}.
  \end{align}
\end{cprop}
\noindent This result allows us to describe Zhu's algebra for the simple admissible level \voas{} $\AdmMod{u}{v}$ of $\AKMA{sl}{2}$ in terms of that of the universal \voas{} $\UVOA{k}$ and the generators (singular vectors) of their maximal ideals.  The next result determines the latter Zhu algebras.
\begin{cprop}[\protect{\cite[Thm.~3.1.1]{FreVer92}}] \label{prop:zhusl2}
  Zhu's algebra for the universal \(\AKMA{sl}{2}\) \voa{} at level \(k \neq -2\) is isomorphic to the \uea{} of non-affine \(\SLA{sl}{2}\):
  \begin{equation}
    \zhu{\UVOA{k}}\cong \UEA\:\SLA{sl}{2}.
  \end{equation}
  Moreover, the image of $T$ in $\zhu{\UVOA{k}}$ may be identified with the quadratic Casimir of \(\SLA{sl}{2}\): 
  \begin{equation}
    T = \frac{1}{2(k+2)}Q,\qquad Q=\frac{1}{2}h^2-ef-fe.
  \end{equation}
\end{cprop}

\flushleft

\end{document}